\documentclass[12pt, reqno]{amsart}

\usepackage{stackrel}
\usepackage{amsmath, amssymb, amsthm, amsfonts}
\usepackage{graphicx}
\usepackage{verbatim}
\usepackage{natbib}
\usepackage{enumitem}
\usepackage{epigraph}
\usepackage{multicol} 

\usepackage[xcharter]{newtxmath}

\usepackage{caption}
\usepackage{subcaption}

\usepackage{fancyvrb}
\usepackage[dvipsnames,svgnames,table]{xcolor}
\usepackage{mdwlist}

\usepackage[breaklinks=true, colorlinks=true, citecolor=Blue, linkcolor=blue]{hyperref}
\newcommand\myshade{85}
\colorlet{mylinkcolor}{violet}
\colorlet{mycitecolor}{Brown}
\colorlet{myurlcolor}{violet}
\hypersetup{
  linkcolor  = mylinkcolor!\myshade!black,
  citecolor  = mycitecolor!\myshade!black,
  urlcolor   = myurlcolor!\myshade!black,
  colorlinks = true,
}
\usepackage[capitalise]{cleveref}
\crefname{theorem}{Theorem}{Theorems}
\crefname{corollary}{Corollary}{Corollaries}
\crefname{lemma}{Lemma}{Lemmas}
\crefname{proposition}{Proposition}{Propositions}
\crefname{definition}{Definition}{Definitions}
\crefname{example}{Example}{Examples}
\crefname{remark}{Remark}{Remarks}
\crefname{assumption}{Assumption}{Assumptions}

\usepackage{tikz}
\usepackage{tikz-cd}
\usetikzlibrary{positioning}
\usetikzlibrary{arrows}
\usetikzlibrary{calc}
\usetikzlibrary{intersections}
\usetikzlibrary{matrix}
\usetikzlibrary{decorations}
\usepackage{pgf}
\usepackage{pgfplots}
\pgfplotsset{compat=1.16}
\usetikzlibrary{shapes, fit}
\usetikzlibrary{arrows.meta}
\usetikzlibrary{decorations.pathreplacing}  
\usetikzlibrary{calligraphy}

\usepackage{pifont}

\usepackage{enumitem}
\setlist[enumerate]{itemsep=2pt,topsep=3pt}
\setlist[itemize]{itemsep=2pt,topsep=3pt}
\setlist[enumerate,1]{label={\upshape (\roman*)}}

\usepackage{bbm}

\usepackage[left=1.25in, 
            right=1.25in, 
            top=1.0in, 
            bottom=1.15in, 
            includehead, includefoot]{geometry}

\newcommand{\navy}[1]{\textit{\textcolor{NavyBlue}{#1}}}
\renewcommand{\leq}{\leqslant}
\renewcommand{\geq}{\geqslant}

\usepackage[ruled, linesnumbered]{algorithm2e}

\DeclareMathOperator{\fix}{fix}

\newcommand*\diff{\mathop{}\!\mathrm{d}}

\newcommand{\zZ}{\mathscr Z}

\newcommand{\RR}{\mathbbm R}

\newcommand{\NN}{\mathbbm N}

\newcommand{\EE}{\mathbbm E}

\newcommand{\Esf}{\mathsf E}

\newcommand{\Xsf}{\mathsf X}

\newcommand{\Zsf}{\mathsf Z}

\renewcommand{\phi}{\varphi}
\renewcommand{\epsilon}{\varepsilon}

\usepackage{aliascnt}

\theoremstyle{plain}
\newtheorem{theorem}{Theorem}[section]

\newaliascnt{corollary}{theorem}

\aliascntresetthe{corollary}

\newaliascnt{lemma}{theorem}
\newtheorem{lemma}[lemma]{Lemma}
\aliascntresetthe{lemma}

\newaliascnt{proposition}{theorem}
\newtheorem{proposition}[proposition]{Proposition}
\aliascntresetthe{proposition}

\theoremstyle{definition}

\newtheorem{assumption}{Assumption}[section]

\AddToHook{env/theorem/begin}{\crefalias{section}{theorem}}
\AddToHook{env/corollary/begin}{\crefalias{section}{corollary}}
\AddToHook{env/lemma/begin}{\crefalias{section}{lemma}}
\AddToHook{env/proposition/begin}{\crefalias{section}{proposition}}
\AddToHook{env/definition/begin}{\crefalias{section}{definition}}
\AddToHook{env/axiom/begin}{\crefalias{section}{axiom}}
\AddToHook{env/example/begin}{\crefalias{section}{example}}
\AddToHook{env/remark/begin}{\crefalias{section}{remark}}
\AddToHook{env/notation/begin}{\crefalias{section}{notation}}
\AddToHook{env/assumption/begin}{\crefalias{section}{assumption}}
\AddToHook{env/condition/begin}{\crefalias{section}{condition}}

\begin{document}

\title{Faithful Decoding}

\author{Nisha Peng}
\address{Research School of Economics, Australian National University}
\email{nisha.peng@anu.edu.au}

\author{John Stachurski}
\address{Research School of Economics, Australian National University}
\email{john.stachurski@anu.edu.au}

\author{Jingni Yang}
\address{School of Economics, University of Sydney}
\email{jingni.yang@sydney.edu.au}

\author{Ziyue Yang}
\address{Research School of Economics, Australian National University}
\email{humphrey.yang@anu.edu.au}

\thanks{The authors thank Shu Hu, Qingyin Ma, Fabio Angelo Maccheroni, Massimo Marinacci, Longye Tian, Thomas J. Sargent, and Shenghao Zhu for many valuable comments and suggestions.}

\begin{abstract}
    This paper studies transformations that increase efficiency in solving equilibrium
    systems without information loss. Our approach
    exploits order-theoretic structure commonly found in economic problems to
    obtain conditions under which high-dimensional systems can be transformed
    into low-dimensional systems while preserving exact relationships
    between their solutions. The transformations can also be used for 
    purposes other than dimensionality reduction, such as simplifying analysis
    and facilitating stochastic approximation routines.
    The theoretical ideas are illustrated using applications from economics and
    finance.  In a real option problem, we demonstrate speed gains of
    up to 70,000 times.
\end{abstract}

\date{July 18, 2026}

\maketitle

\section{Introduction}

Computational methods hold great promise for improving 
understanding of economic systems. Computationally intensive approaches to
economic modeling are enabling the addition of fine-grained heterogeneity
across agents, as well as many forms of adaptive behavior and behavioral biases.
In doing so, these approaches are
progressively shedding light on high-dimensional equilibrium objects, such
as wealth, consumption, and firm size distributions, as well as emergent
phenomena and aggregate implications of micro-level interactions
(see, e.g., \cite{gu2024global}, \cite{payne2025deep}, \cite{yang2025structural}, 
and \cite{wibault2026recurrent}).

A core challenge for researchers is handling the growing complexity of these
large-scale and increasingly high-dimensional models. A common strategy is to
seek low-dimensional representations of high-dimensional equilibrium objects.
One approach is discretization and state aggregation (see, e.g.,
\cite{krusell1998income}, \cite{farmer2017discretizing} or \cite{li2025error}).
Other methods that reduce computational burden include parameterization and power-series methods 
for invariant manifolds, adaptive sparse-grid interpolation, 
projection and collocation methods based on polynomial or spline bases, 
endogenous-grid methods for distributional dynamics (see, e.g., \cite{gomis2009geometric},
\cite{brumm2017using}, \cite{pohl2018higher}, \cite{bayer2024shocks}, or
\cite{ocampo2024computing}). These strategies are broadly applicable and
indisputably important. At the same time, they are all ``lossy:'' information is
degraded in the transformation to a lower-dimensional system.  In applying them,
researchers trade accuracy for compression and computational performance.

In contrast, this paper examines pairs of transformations of dynamical systems
and equilibrium objects that reduce dimensionality in optimization and
equilibrium problems while also being exact.  In particular, the first
transformation, which we refer to as the \emph{encoder}, produces a
lower-dimensional system where computation is efficient and accurate. The
reverse transformation, which we call the \emph{decoder}, transforms solutions
to the lower-dimensional system back into exact solutions for the
higher-dimensional system (modulo floating point error). We also provide
conditions on the encoder and decoder under which a sequence of approximations
to the solution in the low-dimensional system converges, after a transformation,
to a solution to the high-dimensional system.

The methods described here are not unambiguous improvements over the lossy
methods described above (orthogonal projection, state aggregation, etc.),  since
the set of potential applications is narrower. At the same time, there are many
useful settings where perfect decoding is achievable.  Moreover, the
transformations considered here are useful for a range of additional purposes,
beyond dimensionality reduction. For example, in Section~\ref{ss:hb}, we study a
nonlinear system that arises in the theory of asset pricing with heterogeneous
beliefs. One approach to solving the system is stochastic approximation, which
is increasingly used in computational science to solve systems in the presence
of high dimensionality, limited information, or both (see, e.g.,
\cite{mnih2015human}). For the original system, sampled updates are
upward-biased due to a ``max-of-sample-means'' issue \citep{hasselt2010double}.
We show that a transformation of the system makes stochastic approximation
asymptotically convergent. We can then recover the solution to the original
problem from the asymptotic limit of the stochastic approximation scheme, after
reversing the original transformation.

Another example of the value of our theory is given in Section~\ref{s:egf},
where we treat a real option problem.  We show that, in addition to
dimensionality reduction, the encoder shifts us to a new system where the
associated fixed point operator is a contraction mapping, even though the
original operator is not. This demonstrates that the transformations we study
can produce valuable analytical results as well.

Our encoder-decoder method relies on transformations of operators
and, as such, is related to topological conjugacy of dynamical systems. Operators $S$
and $T$ are said to be topologically conjugate if there exists a homeomorphism
$E$ from the domain of $S$ to the domain of $T$ such that $E \circ S = T \circ
E$. When this holds, the fixed points of $S$ and $T$ are in one-to-one
correspondence, and local and global stability properties of fixed points are
also transferred (see, e.g., \cite{devries2014topological}, or
\cite{devaney2018introduction}). 

While topological conjugacy provides beautiful and important results, it does not
offer a practical methodology for dimensionality reduction: Brouwer's
invariance-of-domain theorem (see, e.g., Theorem 62.3 of \cite{munkres2000topology}) implies that $\RR^n$
admits no continuous injective map into $\RR^m$ when $n > m$. To circumvent this problem, we may instead
consider weaker versions of conjugacy, such as the existence of maps $E, D$ such that
\begin{equation}\label{eq:oscpo0}
    E \circ S  = T \circ E
    \quad \text{and} \quad
    D \circ T = S \circ D .
\end{equation}
If either $E$ or $D$ is a homeomorphism, then we are back to topological
conjugacy.  If not, we might hope to obtain dimensionality reduction while still
retaining a connection between fixed points. A number of previous studies have
considered such ideas. For example, topological semiconjugacy requires only
surjective mappings in \eqref{eq:oscpo0} while still providing insights into
dynamics (see, e.g., \cite{day2008algorithms}, \cite{devaney2018introduction} or
\cite{redman2022algorithmic}). However, these conditions are insufficient to
uniquely and exactly transfer fixed points back to the higher dimensional
system.  

Here we take a related but alternative approach that drops the need for
bijective transformations while, at the same time, strengthening the condition
in \eqref{eq:oscpo0} to obtain an exact one-to-one relationship between fixed
points, as well as related convergence results. While our strengthening of
condition \eqref{eq:oscpo0} is nontrivial, we show that many interesting
applications satisfy our conditions. 

As well as exact transformations of solutions, we obtain convergence and
stability results when the encoder and decoder are continuous. For example, if
we can establish either local or global stability of fixed points in the
transformed system (which is typically easier to analyze), that same property is
transferred to the decoded fixed point in the original system.

Furthermore, in Section~\ref{s:opts} we develop results for applications with
monotonicity properties.  For example, we show that monotone shifts in
equilibrium objects in low-dimensional systems translate into monotone shifts in
semiconjugate high-dimensional systems. These results are valuable for fields
such as finance and operations research, since, in these application domains,
many problems have natural order-theoretic structure.  For example, in
intertemporal optimization problems, current values of reward streams are
typically (weakly) increasing in future values.  In asset pricing, current
values of assets are typically increasing in expected future values of assets.
In production network problems, prices of goods in one part of the production
chain are typically increasing in the price of goods in the rest of the chain,
and so on.

Regarding related literature, our paper builds on earlier work by authors such as
\cite{bertsekas2012dynamic} and \cite{ma2021dynamic}. These authors study
transformations of optimality equations for dynamic programs in a way that
reduces state space dimension without causing information loss. (The analysis in
\cite{bertsekas2012dynamic} considers integrating out ``uncontrollable states''
to reduce dimensionality.) Similar ideas have been applied in many settings,
including structural estimation \citep{rust1987optimal, hotz1993conditional,
iskhakov2016comment} and optimal stopping \citep{ma2019optimal}. Here we extend
these ideas by constructing a more general and systematic framework for handling
equilibrium and fixed point problems through ``encoding'' and ``decoding.''  Our
convergence and fixed point results extend these ideas to a broader range of
applications.  

The paper proceeds as follows. Section~\ref{s:egf} presents a real option
example that illustrates our main ideas and computational gains.
Section~\ref{s:r} develops the general theory of strong semiconjugacy and fixed
point transfer. Section~\ref{s:fa} provides further applications, including a turn-based zero-sum
stochastic game as well as asset pricing problems, one of which features
heterogeneous beliefs and a stochastic-approximation-based solution method.  Section~\ref{s:opts} extends
our theory by considering parametric monotonicity results. Implementation
details and additional results are given in the appendix.

\section{Example: A Real Option Problem}\label{s:egf}

This section provides a relatively simple example that helps illustrate the key
ideas. The application involves real options, which capture the value of
flexibility to adapt investments amid uncertainty, such as delaying decisions
(value of waiting) until new information resolves risks.\footnote{Recent
    research on real options includes applications to climate adaptation
    pathways for transit infrastructure \citep{martello2024}, agricultural
    investments under policy uncertainty \citep{stroombergen2022}, renewable
    power-to-hydrogen sequencing \citep{fabianek2024}, and market entry
balancing growth and deferral options amid endogenous and exogenous
uncertainties \citep{li2007, fajgelbaum2017uncertainty}.}  Here we focus on a
market entry problem closely related to \cite{fajgelbaum2017uncertainty},
although the same ideas apply to many other real option problems. We write down
an operator such that the fixed point of this operator describes the expected
present value of a firm with an option to enter a market over all states of the
world. Next, we will consider a closely related operator that acts on functions
taking values in a lower-dimensional space.   We will show, in addition, that
the fixed point of the low-dimensional operator can be mapped back to the
higher-dimensional space to solve the original problem.

\subsection{Setup}\label{ss:fv}

Consider a firm that can choose to enter a market and receive a given
payoff or wait until the next period and reassess.
The firm makes its decision at time $t$ after
observing the ``state of the world'', represented by
a joint state $X_t$ that contains all relevant information on
macroeconomic and market-specific shocks.

Given current state $X_t$, which takes values in set $\Xsf$, the payoff to
entering is given by $\pi(X_t)$. The interest rate $r_t = r(X_t)$ evolves
stochastically over time, also depending on the current joint state $X_t$.  Let
$(V_t)_{t \geq 0}$ be the value of the firm at each point in time.  This
sequence obeys the recursion
\begin{equation}\label{eq:vrec}
    V_t = 
    \max
	\left\{
        \pi(X_t), \; -c(X_t) + \frac{1}{1 + r(X_t)} \EE[ V_{t+1} \,|\, X_t]
	\right\},
    \qquad t = 0, 1, \ldots,
\end{equation}
where $c(X_t)$ is the flow cost of continuing to develop the product in state
$X_t$.  If the sequence $(V_t)$ can be computed, the optimal choice of whether
or not to exercise the entry option reduces to choosing the maximal element
on the right-hand side of \eqref{eq:vrec}.  For example, if $\pi(X_t)$ is
larger, the firm should exercise.

In view of the Markov structure, we conjecture that a solution $(V_t)_{t \geq 0}$
to this recursion must have the form $V_t = v(X_t)$ for some real-valued
function $v$.  Working with this conjecture leads us to rewrite \eqref{eq:vrec}
in the functional form
\begin{equation}\label{eq:lf}
    v(x) = \max
	\left\{
        \pi(x),  \; - c(x) + \beta(x) \int v(x') N(x, \diff x') 
	\right\},
\end{equation}
where 
\begin{equation*}
    \beta(x) \coloneq \frac{1}{1 + r(x)}  
    \qquad (x \in \Xsf),
\end{equation*}
is the state-dependent discount factor and $N$ is the stochastic kernel for the
state process $(X_t)_{t \geq 0}$.
We understand \eqref{eq:lf} as a functional equation in unknown function $v$.
    
A natural way to solve for the lifetime value function $v$ is to introduce an operator $S$ 
that maps $v$ to $Sv$ via
\begin{equation*}
	(Sv)(x) = \max
        \left\{
            \pi(x), \; -c(x) + \beta(x) \, \int v(x') N(x, \diff x')
        \right\}.
\end{equation*}
Clearly $v$ is a solution to \eqref{eq:lf} if and only if $v$ is a fixed point of
$S$.  While we can attempt to solve for a fixed point of $S$ directly, there are
settings where we can solve the problem more efficiently, by reducing
dimensionality.

To illustrate one such setting, suppose now that the state process $X_t$
consists of both persistent and transient components, as is common in
applications. In particular, suppose that $X_t$ has the form $X_t = (Z_t, \xi_t)$, where
\begin{itemize}
    \item $(Z_t)_{t \geq 0}$ is a persistent vector-valued process driven by
        stochastic kernel $Q$, so that $Q(z, \cdot)$ is the distribution of
        $Z_{t+1}$ given $Z_t = z$, and
    \item $(\xi_t)_{t \geq 0}$ is a transient vector-valued process, independent of $(Z_t)_{t \geq 0}$.
\end{itemize}
We suppose that the transient shock sequence $(\xi_t)_{t \geq 0}$ is {\sc iid}
with common distribution $\phi$.
The sequence $(Z_t)$ takes values in $\Zsf \subset \RR^k$, while $(\xi_t)$ takes
values in $\Esf \subset \RR^\ell$. The set $\Xsf$ is the product space $\Zsf \times
\Esf$. The stochastic kernel $N$ for $(X_t)$ is connected to the new primitives
$Q$ and $\phi$ via $N(x, A \times B) = N((z, \xi), A \times B) =  Q(z, A) \phi(B)$.

We now introduce a second operator $T$ via
\begin{equation}\label{eq:fvt}
    (T w)(z) = \int \max
    \left\{
    \pi(z, \xi), \;  -c(z, \xi) + \beta(z, \xi) \, \int w(z') Q(z, \diff z')
    \right\} \phi(\diff \xi).
\end{equation}
Notice that, unlike $S$, the operator $T$ acts on real-valued functions defined
on $\Zsf$, rather than on real-valued functions defined on the larger space
$\Xsf = \Zsf \times \Esf$.  

To help us interpret the action of the operator $T$, we also introduce a related
map $D$ via
\begin{equation}\label{eq:cdg}
    (D w)(z, \xi) = \max
    \left\{
        \pi(z, \xi), \; -c(z, \xi) + \beta(z, \xi) \, \int w(z') Q(z, \diff z')
    \right\}.
\end{equation}
We call $D$ the \navy{decoder map} for the pair $S,T$.  We show below that,
under some mild discounting and regularity assumptions, the following results
are true:
\begin{enumerate}
    \item The map $T$ has a unique fixed point $\bar w$
        within a suitable class of real-valued functions over $\Zsf$.
    \item As a \emph{consequence} of this fact, the map $S$ has a unique fixed point 
        within a suitable class of real-valued functions over $\Xsf$.
    \item The \navy{decoded image} $\bar v \coloneq D \bar w$ of the fixed point $\bar w$ of $T$ 
        is exactly the unique fixed point  of $S$ (and hence the unique solution to the
        real option problem).
    \item For any $w$ from a large class of initial conditions, the sequence
        $T^j w$ converges to the fixed point $\bar w$ of $T$.
    \item For any such $w$, the \navy{decoded sequence} $D T^j w$ converges to
        the unique fixed point $\bar v$ of $S$.
\end{enumerate}
There are many useful implications of these results.  One is that, for suitably
chosen $w$, we can iterate with the low-dimensional operator $T$ and then decode
the last iterate to approximate the solution to the firm problem at any level of
accuracy. The next section helps to illustrate the advantages of this approach.
The theoretical results discussed above are established in Section~\ref{s:r}.

\subsection{A Computational Illustration} \label{ss:comp_fv2}

To study efficiency gains from iterating with $T$ instead of $S$, we now
evaluate iteration for these two cases under a specific parameterization.
We suppose that the profit function is given by
\begin{equation}\label{eq:profit_function}
    \pi(z, \xi) = \max_{n \geq 0} 
    \left\{
        p(z, \xi) a(z, \xi) n^\theta - \omega(z, \xi) n
    \right\},
\end{equation}
where $p$ is the output price, $\omega$ is the wage rate, $a$ is a productivity term,
$n$ is labor input, $\theta \in (0,1)$ is a production elasticity, $z \in \RR^k$
is a vector of persistent shocks and $\xi \in \RR^\ell$ is a vector of transient
components. For the purpose of this experiment, we set $k = \ell = 4$.\footnote{In
the computations, the discount factor $\beta(z, \xi)$ is rescaled so that every
realization lies in $(\beta_{\min}, \beta_{\max}) \subset (0,1)$.  For the
functional forms we set $\beta(z_0, \xi_0) = z_0 + \xi_0$, $p(z_1, \xi_1) = z_1 +
\xi_1$,  $\omega(z_2, \xi_2) = z_2 + \xi_2$, $a(z_3, \xi_3) = z_3 + \xi_3$.}
The flow cost $c(z, \xi)$ is set to zero.

We set the persistent components to $Z_i = \exp(Y_i)$, where
$Y_i$ follows an independent Gaussian AR(1) process discretized by the
Tauchen method \citep{tauchen1986finite}. The pre-discretization dynamics follow
$Y_i' = \rho_i Y_i + \nu_i \varepsilon_i'$, where primes denote next period
values and each $\varepsilon_i'$ is an independent standard normal draw.
The transient shocks $\xi_i$ are {\sc iid} and drawn from a lognormal density
with parameters $(\mu_i, \sigma_i)$. We discretize the state space using a grid
with $n_z = (n_0, n_1, n_2, n_3)$ and $m_\xi = (m_0, m_1, m_2, m_3)$ points for
each component of $Z_t$ and $\xi_t$, respectively (e.g., the first component of
$Z_t$ is discretized onto a grid of size $n_0$).
\begin{figure}
    \centering
    \includegraphics[width=0.8\textwidth]{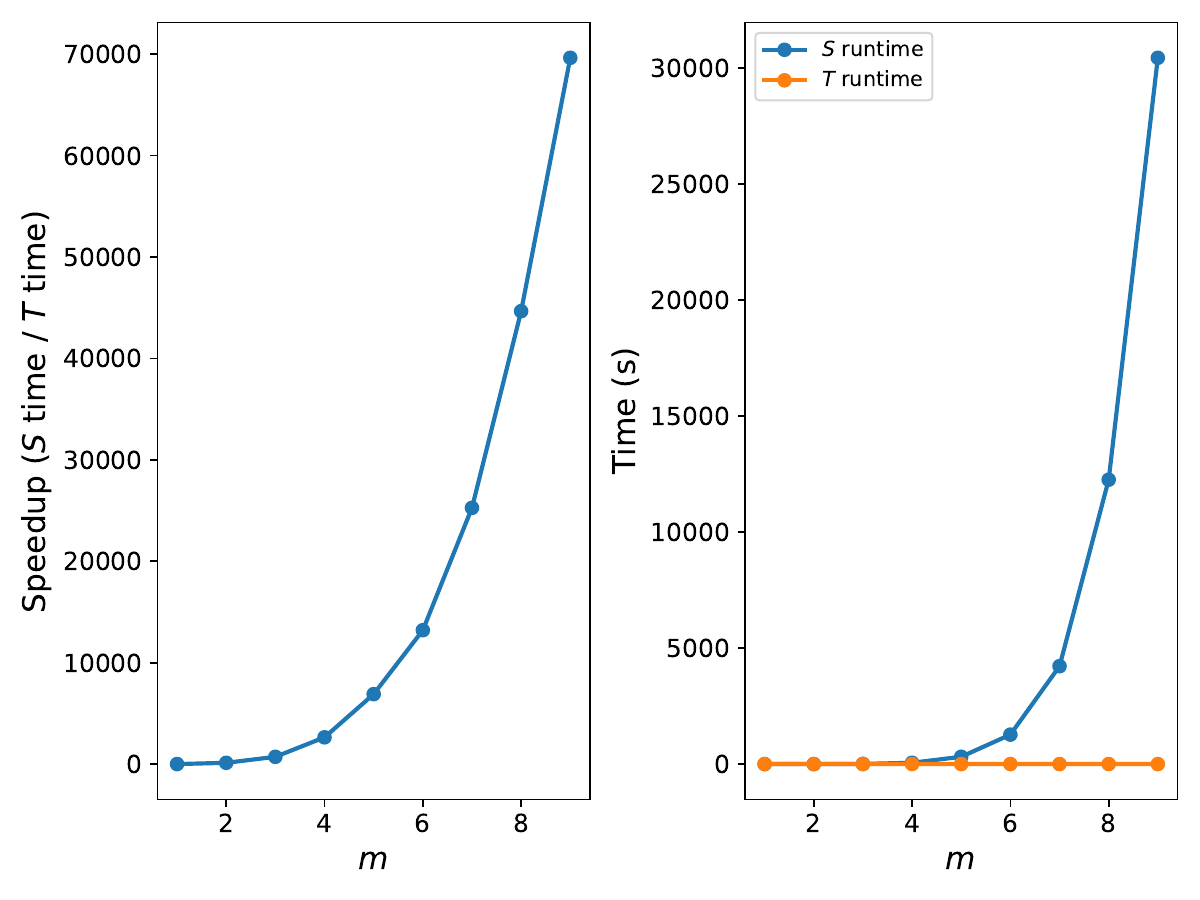}
    \caption{Left: Speed gain obtained by iterating with $T$ instead of $S$.
        Right: Raw computation times for the fixed points of $S$ and $T$ as a
        function of the common transient grid size $m_\xi= (m, m, m, m)$.}
    \label{fig:numba_speedup}
\end{figure}

\begin{table}
    \centering
    \caption{Performance results and accuracy validation (times in seconds)}
    \label{tab:results}
    \begin{tabular}{lrrrrr}
    \hline
    grid size $m$ & $t_S$ ($S$ time) & $t_T$ ($T$ time) & $t_S/t_T$ (speed gain) & $r_S$ & $r_T$ \\
    \hline
    1  & 0.012 & 0.0026 & 4.6 & $9.9 \times 10^{-7}$ & $9.6 \times 10^{-7}$ \\
    2  & 0.38 & 0.0030 & 128 & $7.1 \times 10^{-7}$ & $8.8 \times 10^{-7}$ \\
    3  & 6.14 & 0.0086 & 719 & $7.6 \times 10^{-7}$ & $7.9 \times 10^{-7}$ \\
    4  & 55.9 & 0.021 & 2,641 & $5.5 \times 10^{-7}$ & $6.8 \times 10^{-7}$ \\
    5  & 315 & 0.046 & 6,898 & $6.6 \times 10^{-7}$ & $8.6 \times 10^{-7}$ \\
    6  & 1,269 & 0.096 & 13,201 & $9.5 \times 10^{-7}$ & $7.1 \times 10^{-7}$ \\
    7  & 4,219 & 0.167 & 25,267 & $9.9 \times 10^{-7}$ & $6.8 \times 10^{-7}$ \\
    8  & 12,258 & 0.274 & 44,656 & $5.7 \times 10^{-7}$ & $6.0 \times 10^{-7}$ \\
    9  & 30,443 & 0.437 & 69,632 & $7.5 \times 10^{-7}$ & $6.0 \times 10^{-7}$ \\
    \hline
    \end{tabular}
\end{table}

\noindent
In all experiments we compare direct iteration on $S$ with iteration on $T$
using the same stopping tolerance $\varepsilon=10^{-6}$. Starting from fixed
initial iterates $v_0$ and $w_0$, for the high-dimensional problem we generate
$v_{q+1}=Sv_q$ until
$\lVert v_{q+1}-v_q\rVert_\infty<\varepsilon$; we record the runtime $t_S$ and
set $\hat v=v_{q+1}$. We repeat the same procedure for the low-dimensional
problem, generating $w_{j+1}=Tw_j$ until
$\lVert w_{j+1}-w_j\rVert_\infty<\varepsilon$; this yields runtime $t_T$ and
terminal iterate $\hat w=w_{j+1}$. The reported speedup is $t_S/t_T$. As
accuracy diagnostics we report the terminal successive-iterate differences
\begin{equation}\label{eq:residuals}
    r_S\coloneq\lVert v_{q+1}-v_q\rVert_\infty,
    \qquad
    r_T\coloneq\lVert w_{j+1}-w_j\rVert_\infty.
\end{equation}
Throughout, $n_i=4$ and $m_i=m$ for all $i$. Figure~\ref{fig:numba_speedup} and
Table~\ref{tab:results} report runtimes, speedups, and the accuracy diagnostics
in \eqref{eq:residuals}; Appendix~\ref{a:params} lists the parameter values.

The results show large speed gains that grow polynomially with $m$. For larger
$m$, $S$ operates on the full transient-shock grid
$(m_0\times m_1\times m_2\times m_3)=m^4$ and repeatedly integrates over
successor shocks, whereas $T$ grows more slowly because it separates the $z'$
expectation from the $\xi$ integration. Consequently, the speedup grows roughly
quartically with $m$, as shown in Figure~\ref{fig:numba_speedup} and
Table~\ref{tab:results}. At $m=4$, iterating with $T$ is already three orders of
magnitude faster than iterating with $S$. At $m=9$, it is around 70,000 times
faster under the same stopping tolerance.

Later, in the theory section, we show that the fixed point $\bar v$ of $S$ is
the decoding $D\bar w$ of the low-dimensional fixed point $\bar w$. Hence, with
$\hat v$ and $\hat w$ approximating $\bar v$ and $\bar w$, respectively, we
should have $\hat v\approx D\hat w$. To verify this, we compute the sup-norm
discrepancy
$$
    \delta \coloneq \max_{z,\xi} \bigl|D\hat w(z,\xi) - \hat v(z,\xi)\bigr|.
$$  
Across all grid sizes, $\delta<9.9\times10^{-7}$, indicating that the two
numerical fixed points agree within the stopping tolerance.

\section{Theory}\label{s:r}

This section provides a theoretical foundation for the transformation
discussed in Section~\ref{s:egf}, as well as for a range of additional 
applications discussed below.  We begin describing the theory with an informal overview in Section~\ref{ss:over}, and then move on to a more formal discussion,
starting in Section~\ref{ss:prel}.

\subsection{Overview}\label{ss:over}

Suppose that we are interested in the solution to a finite- or
infinite-dimensional system of potentially nonlinear equations. Let $V$ be a set
of candidate solutions and let $S \colon V \to V$ be an operator such that $v
\in V$ is a solution to the system if and only if $v$ is a fixed point of $S$.
We seek a potentially lower-dimensional set $W$, an operator $T \colon W \to W$,
and a set of conditions on these operators under which the high-dimensional
operator $S$ has a unique fixed point in $V$ if and only if the low-dimensional
operator $T$ has the same property, and under which any solution to the
lower-dimensional fixed point problem can be transformed back into an exact
solution for the higher-dimensional system.

The transformations we consider involve a decomposition.  In particular, we seek
maps $E \colon V \to W$ and $D \colon W \to V$ such that the
high-dimensional operator $S$ and the low-dimensional operator $T$ obey
\begin{equation}\label{eq:cojm0}
    S = D \circ E \text{ on } V
    \qquad \text{and} \qquad 
    T = E \circ D \text{ on } W.
\end{equation}
We show that this structure is enough
to exactly match fixed points across the low- and high-dimensional systems. We
also show that, under additional continuity conditions on $E$ and $D$,
stability properties of fixed points are also preserved: stability of a fixed
point in one system implies stability in the other.

Notice that \eqref{eq:cojm0} implies the semiconjugate relationship in \eqref{eq:oscpo0}.
Hence our approach is, in essence, a specialized version of semiconjugacy. 
As shown below, the strengthening of \eqref{eq:oscpo0}
to \eqref{eq:cojm0} allows us to retain and transfer exact qualitative and
quantitative fixed point information, as well as, in some settings, convergence results for approximation
methods.\footnote{Requiring only \eqref{eq:oscpo0} is not enough to deliver the
    same results.  For instance, consider sets $V$ and $W$, each with more than one element.
    Let $S$ be the identity map on $V$ and let $T$ be a self-map on $W$ which
    maps every $w \in W$ to the same point $w^* \in W$. 
    If we let $E \colon V \to W$ map every $v \in V$ to $w^*$ and $D \colon W
    \to V$ map every $w \in W$ to the same point $v^* \in V$, then
    \eqref{eq:oscpo0} holds. At the same time,   every $v \in V$ is a fixed
point for $S$ and $w^*$ is the only fixed point for $T$.}

We now formalize the ideas sketched above. We begin by
recalling some standard definitions from the theory of dynamical systems and
then introduce the notion of strong semiconjugacy.  After that, we state and
prove our main theoretical results on fixed point transfer and stability.

\subsection{Preliminaries}\label{ss:prel}

To proceed further, we first recall some useful concepts and definitions.
Let $V$ be a topological space. A \navy{fixed point} of a self-map $S$
on set $V$ is a $v \in V$ with $Sv=v$. 
Throughout this paper, the notation $\fix(S)$ refers to the set $\left\{v \in V \colon Sv = v \right\}$.
A \navy{dynamical system} is a pair $(V,
S)$ where $V$ is a set and $S$ is a self-map on $V$.  The system $(V, S)$ is
called \navy{globally stable} when $S$ has a unique fixed point $\bar v$ in $V$
and $S^n v \to \bar v$ as $n \to \infty$ for all $v \in V$.
A fixed point $\bar v$ of $S$ in $V$ is called \navy{locally stable} for $(V, S)$
if there exists a neighborhood $\hat V \subseteq V$ of $\bar v$ such that
$S^n v \to \bar v$ for all $v \in \hat V$.

Let $(V, S)$ and $(W, T)$ be dynamical systems, where $V$ and $W$ are
topological spaces. We call $(V, S)$ and $(W, T)$ \navy{strongly semiconjugate under
$E,D$} when there exist maps $E \colon V \to W$  and $D \colon
W \to V$ such that $S = D \circ E$  on $V$ and $T = E \circ D$ on $W$; that is,
such that \eqref{eq:cojm0} holds.  We call $E$ the \navy{encoder map} and $D$
the \navy{decoder map}. The ``strongly semiconjugate'' terminology comes from the fact
that, when \eqref{eq:cojm0} holds, so do the two equalities in
\eqref{eq:oscpo0}. 

Figure~\ref{fig:conj_semiconj} illustrates the difference between conjugacy,
where the encoder map $E$ is a bijection, and strong semiconjugacy.

\begin{figure}[ht]
    \centering
    \begin{minipage}[t]{0.48\textwidth} 
    \centering
    \caption*{\small Conjugacy}
    \vspace{0.25cm}
    \begin{tikzcd}[column sep=4em, row sep=4em] 
        V \arrow[r, "S"]
        \arrow[d, "E"', bend right=20, pos=0.51] 
        & V \arrow[d, "E", bend left=20, pos=0.51] \\ 
        W \arrow[r, "T"']
        \arrow[u, "E^{-1}"', bend right=20, pos=0.5] 
        & W \arrow[u, "E^{-1}", bend left=20, pos=0.5] 
    \end{tikzcd}
    \end{minipage}
    \begin{minipage}[t]{0.48\textwidth} 
        \centering
        \caption*{\small Strong semiconjugacy}
        \vspace{0.25cm}
    \begin{tikzcd}[column sep=4em, row sep=4em] 
        V \arrow[r, "S"]
        \arrow[d, "E"', bend right=20, pos=0.51] 
        & V \arrow[d, "E", bend left=20, pos=0.51] \\ 
        W \arrow[r, "T"']
        \arrow[u, "D"', bend right=20, pos=0.5] 
        & W \arrow[u, "D", bend left=20, pos=0.5] 
    \end{tikzcd}
    \end{minipage}
    \caption{Comparison of conjugacy and strong semiconjugacy}
    \label{fig:conj_semiconj}
\end{figure}

In the applications we study, a typical scenario has the following form: we
study a dynamical system $(V, S)$ and we wish to (a) show that $S$ has a unique
fixed point in $V$ and (b) propose a method for computing that fixed point.
Rather than tackle this problem directly, we aim to work with a strongly
semiconjugate system $(W, T)$, presumably motivated by the fact that $(W, T)$
is, for some reason, easier to work with than $(V, S)$. In particular, we aim to 
\begin{itemize}
    \item deduce the existence of a unique fixed point of $(V,S)$ only by
        studying $(W, T)$, and
    \item compute this unique fixed point, working only with the ``nice'' system
        $(W, T)$.
\end{itemize}
The theoretical results presented in the next section will help us execute this
and other closely related tasks.

\subsection{Theoretical Results}

Our first result provides foundations for the approach outlined above. 
In the statement of the theorem, $(V, S)$ and $(W, T)$ are arbitrary dynamical
systems.

\begin{theorem}\label{p:foo0}
    If $(V, S)$ and $(W, T)$ are strongly semiconjugate under $E, D$,
    then the following statements are true:
    \begin{enumerate}
        \item If $\bar v$ is a fixed point of $S$ in $V$, then $E\bar v$ is a fixed point
            of $T$ in $W$.
        \item If $\bar w$ is a fixed point of $T$ in $W$, then $D \bar w$ is a
            fixed point of $S$ in $V$.
        \item $\fix(S)$ and $\fix(T)$ have the same cardinality.
    \end{enumerate}
\end{theorem}

\begin{proof}
    Let $(V, S)$ and $(W, T)$ be as stated. If $\bar v$ is a fixed point of
    $S$ in $V$, then $T E \bar v = E S \bar v = E \bar v$, so $E\bar v$ is a fixed point of
    $T$ in $W$.  Similarly, if $\bar w$ is a fixed point of
    $T$ in $W$, then $S D \bar w = D T \bar w = D
    \bar w$, so $D
    \bar w$ is a fixed point of $S$ in $V$. This proves (i)--(ii).

    Regarding (iii), we claim that $\fix(T) = E \fix(S)$. On the one hand, by (i),
    $E \fix(S) \subseteq \fix(T)$. On the other hand, for any given $w \in \fix(T)$,
    $T w = E D w = w$. By (ii), $v = D w \in \fix(S)$,
    and so $w \in E \fix(S)$ by definition. Hence, $E \colon \fix(S) \to \fix(T)$
    is well-defined and surjective. Now we show that $E \colon \fix(S) \to \fix(T)$
    is injective. For any given $v_1 \neq v_2 \in \fix(S)$, suppose that 
    $E v_1 = E v_2$. Applying $D$ to both sides yields $S v_1 = S v_2$. Since
    $v_1$ and $v_2$ are fixed points of $S$, this gives $v_1 = v_2$, a contradiction.
    Hence, $E \colon \fix(S) \to \fix(T)$ is bijective, i.e.,
    $\fix(S)$ and $\fix(T)$ have the same cardinality.
\end{proof}

Our next result shows that local stability transfers from $(W, T)$ to $(V, S)$
under strong semiconjugacy, provided that both the encoder and decoder maps are
continuous. 

\begin{theorem}\label{t:local}
    Let $(V, S)$ and $(W, T)$ be dynamical systems where $V$ and $W$ are both
    topological spaces. If
    \begin{enumerate}
        \item $(V, S)$ and $(W, T)$ are strongly semiconjugate under $E, D$,
        \item the decoder map $D$ is continuous on $W$,
        \item the encoder map $E$ is continuous on $V$, and
        \item $\bar w$ is locally stable for $(W, T)$,
    \end{enumerate}
    then $\bar v \coloneq D \bar w$ is also locally stable for $(V, S)$. Moreover, there
    exists a subset $W_1 \subset W$ such that
    \begin{equation}\label{eq:local}
        \text{for any $w \in W_1$},\qquad
        D T^n w \to \bar v \quad \text{as} \quad n \to \infty.
    \end{equation}
\end{theorem}

\begin{proof}
    Let $(V, S)$ and $(W, T)$ be as stated. By local stability of $\bar w$, there exists
    an open neighborhood $\hat W \subset W$ of $\bar w$ such that $T^n w \to \bar w$ for all $w
    \in \hat W$. By (ii) of \cref{p:foo0}, $\bar v = D \bar w$ is a fixed
    point of $S$. By continuity of $E$, the set $E^{-1}(\hat W) \coloneq \left\{ v \in V \colon E v \in \hat W \right\}$ is open. We have $\bar v \in E^{-1}(\hat W)$ since $E \bar v = E D \bar w = T \bar w = \bar w \in \hat W$.
    Hence, $E^{-1}(\hat W)$ is a neighborhood of $\bar v$. Moreover, for each $v \in E^{-1}
    (\hat W)$, local stability of $\bar w$ and continuity of $D$ yield
    \begin{equation*}
        S^n v = (D E)^n v
        = D (E D)^{n-1} E v = D T^{n-1} E v \to \bar v.
    \end{equation*}
    Hence, $\bar v$ is locally stable for $(V, S)$ and $W_1 = \hat W$ is the subset
    satisfying \eqref{eq:local}.
\end{proof}

Next we provide a global stability result.  It allows us to (a) deduce the
existence of a unique fixed point of $(V,S)$ only by studying the ``nicer''
system $(W, T)$, and (b) compute this unique fixed point, working only with $(W,
T)$.

\begin{theorem}\label{t:sosft}
    Let $(V, S)$ and $(W, T)$ be dynamical systems where $V$ and $W$ are both
    topological spaces. If
    \begin{enumerate}
        \item $(V, S)$ and $(W, T)$ are strongly semiconjugate under $E, D$,
        \item the decoder map $D$ is continuous on $W$, and
        \item $(W, T)$ is globally stable with unique fixed point $\bar w$,
    \end{enumerate}
    then $(V, S)$ is also globally stable, with unique fixed point $\bar v
    \coloneq D \bar w$. Moreover,
    \begin{equation}\label{eq:wwtt}
        \text{for any $w \in W$},\qquad
        D T^n w \to \bar v \quad \text{as} \quad n \to \infty.
    \end{equation}
\end{theorem}

Notice that in \eqref{eq:wwtt} we iterate only in $(W, T)$,
and finally transfer back to $V$ using the mapping $D$.

\begin{proof}[Proof of \cref{t:sosft}]
    Let $(V, S)$ and $(W, T)$ be as stated.
    To see the global stability of $(V, S)$, observe that, since
    $\bar w$ is the unique fixed point of $T$, (iii) of
    \cref{p:foo0} implies that $\bar v = D \bar w$ is the unique
    fixed point of $S$. Fix $v \in V$, global stability of $(W, T)$ and
    continuity of $D$ yield
    \begin{equation*}
        S^n v = (D E)^n v
        = D (E D)^{n-1} E v = D T^{n-1} E v \to \bar v.
    \end{equation*}
    Similarly, fix $w \in W$, global stability of $(W, T)$ and
    continuity of $D$ yield $D T^n w \to \bar v$.
\end{proof}

\subsection{Proofs of the Section~\ref{s:egf} Results} \label{ss:pfegf}

In this section, we return to the setting of Section~\ref{s:egf} and restate and
then prove some of the informal claims made regarding the real option
problem. In the process, we will show how Theorem~\ref{t:sosft} can be used to
confirm that iterating $T$ is an effective way to compute the unique fixed point
of $S$.

Recall from Section~\ref{ss:fv} that $S$ is given by \eqref{eq:lf}, that $T$ is
given by \eqref{eq:fvt}, and that $D$ is given by \eqref{eq:cdg}. To analyze
these maps, let $\ell \Xsf$ denote the space of Borel measurable functions
$h \colon \Xsf \to \RR$, modulo the equivalence relation $h \sim h'$ whenever
$h(z, \cdot) = h'(z, \cdot)$ $\phi$-a.e.\ for every $z \in \Zsf$, satisfying
\begin{equation}\label{eq:lbo}
    \sup_{z \in \Zsf} \int |h(z, \xi)| \phi(\diff \xi) < \infty,
\end{equation}
and let $b\Zsf$ denote the set of bounded Borel measurable functions in
$\RR^\Zsf$. We equip $\ell \Xsf$ with the norm \eqref{eq:lbo} and $b\Zsf$ with
the supremum norm $\|\cdot\|_\infty$; under \cref{a:fv} below, $S \colon \ell \Xsf \to \ell \Xsf$
and $T \colon b\Zsf \to b\Zsf$.
Let $E$ be the encoder map on $\Zsf$ defined by
\begin{equation}\label{eq:cdf}
    (Ev)(z) = \int v(z, \xi) \phi(\diff \xi).
\end{equation}

\begin{assumption}\label{a:fv}
	The following conditions hold:
	\begin{enumerate}
        \item the payoff $\pi$ and flow cost $c$ satisfy $\pi, c \in \ell\Xsf$,
        \item the state-dependent discount factor $\beta \colon \Xsf \to (0,
            \infty)$ is bounded and Borel measurable,
		\item there exists an $n \in \NN$ such that
        \begin{equation*}
			\sup_{z \in \Zsf} \EE_z^Q \prod_{t=0}^{n-1} \hat \beta(Z_t) < 1,
		\end{equation*}
        where $\hat \beta(z) \coloneq \int \beta(z, \xi) \phi(\diff \xi)$ and $\EE^Q$ is
        the conditional expectation operator with respect to stochastic kernel $Q$.
	\end{enumerate}
\end{assumption}

We can now state our main results for the real option problem. 

\begin{proposition}\label{p:fk}
    If \cref{a:fv} holds, then the real option problem has
    a unique solution $\bar v$ in $\ell \Xsf$. Moreover, for any $w \in b\Zsf$, we
    have
    \begin{equation}\label{eq:gtgt}
        D T^n w \to \bar v
        \quad \text{as} \quad
        n \to \infty.
    \end{equation}
\end{proposition}

The proof of Proposition~\ref{p:fk} is given in stages below.
The existence and uniqueness claim is established by (i)
proving that the simple system $(b\Zsf, T)$ has a unique fixed point via an
eventual contraction property and (ii) translating this unique fixed point to a unique
fixed point for $S$ via the strong semiconjugacy described in Lemma~\ref{l:stf}.
Then we use continuity of $D$ and strong semiconjugacy to establish the
convergence results in \eqref{eq:gtgt}. These convergence results provide an
efficient method for computing $\bar v$, since all iteration takes place in the
``nice'' system $(b\Zsf, T)$.

The next lemma describes the domain and range of $E$ and $D$.

\begin{lemma}\label{l:gfmt}
    The maps $E$ and $D$ obey $E \colon \ell \Xsf \to b\Zsf$ and $D \colon b\Zsf \to \ell \Xsf$.
\end{lemma}

\begin{proof}
    To see that $E \colon \ell \Xsf \to b\Zsf$, fix $v \in \ell \Xsf$.  Using
    the definition of $\ell\Xsf$, we have 
    \begin{equation*}
        \sup_{z \in \Zsf} |(Ev)(z)|
        = \sup_{z \in \Zsf} |\int v(z, \xi) \phi(\diff \xi)| 
        \leq \sup_{z \in \Zsf} \int |v(z, \xi)| \phi(\diff \xi) < \infty.
    \end{equation*}
    In particular, $Ev \in b\Zsf$.

    To see that $D \colon b\Zsf \to \ell \Xsf$, fix $w \in b\Zsf$.  The function
    \[
        -c(z, \xi) + \beta(z, \xi) \int w(z') Q(z, \diff z')
    \]
    is the sum of $-c \in \ell\Xsf$
    and a bounded function, hence it is in $\ell \Xsf$.
    Note that $\pi \in \ell \Xsf$. Since $Dw = \max\{\pi, \, -c + \beta \int w \, \diff Q\}$,
    we have $D w \in \ell \Xsf$.
\end{proof}

We use $E$ and $D$ to connect $S$ and $T$ as follows:

\begin{lemma}\label{l:stf}
    The systems $(\ell \Xsf, S)$ and $(b\Zsf, T)$ are strongly semiconjugate under $E,D$.
\end{lemma}

\begin{proof}
    Fixing $v \in \ell \Xsf$, we have
    \begin{align*}
        (D E v)(z, \xi) &= \max
            \left\{
            \pi(z, \xi), \;  -c(z, \xi) + \beta(z, \xi) \int (E v)(z') Q(z, \diff z')
            \right\}
            \\
            &= \max
            \left\{
                \pi(z, \xi), \; -c(z, \xi) + \beta(z, \xi) \int \int v(z', \xi') \phi(\diff \xi') Q(z, \diff z')
            \right\}
            \\
            &= (Sv)(z, \xi).
    \end{align*}
    Moreover, for each $w \in b\Zsf$, we have
    \begin{align*}
        (E D w)(z) &= \int (D w)(z, \xi) \phi(\diff \xi)
        \\
        &= \int \max
        \left\{
            \pi(z, \xi), \; -c(z, \xi) + \beta(z, \xi) \int w(z') Q(z, \diff z')
        \right\} \phi(\diff \xi)
        \\
        &= (Tw)(z).
    \end{align*}
    This confirms that $S = D \circ E$ on $\ell\Xsf$ and $T = E \circ D$ on $b\Zsf$.
\end{proof}

\begin{lemma}\label{l:gs}
    If \cref{a:fv} holds, then $(b\Zsf, T)$ is globally stable.
\end{lemma}

\begin{proof}
    Fix $w_1, w_2 \in b\Zsf$. For each $z \in \Zsf$, we have
    \begin{align*}
        |(T w_1)(z) - (T w_2)(z)|
	& \leq \int \beta(z, \xi) \phi(\diff \xi) 
        \int |w_1(z') - w_2(z')| Q(z, \diff z')\\
        & = \hat \beta(z) \int |w_1(z') - w_2(z')| Q(z, \diff z').
    \end{align*}
    With condition (iii) of \cref{a:fv},
    Theorem~2.1 of \cite{stachurski2021dynamic} implies that $T$ is eventually
    contracting and hence globally stable.
\end{proof}

\begin{lemma} \label{l:gcont}
    The decoder $D \colon b\Zsf \to \ell\Xsf$ defined as \eqref{eq:cdg} is a continuous map.
\end{lemma}

\begin{proof}
    Fix $w_n \to w$ in $b\Zsf$, so that $\| w_n - w \|_\infty \to 0$.
    For each $(z, \xi) \in \Xsf$, we have
    \begin{equation*}
        |(D w_n)(z, \xi) - (D w)(z, \xi)|
        \leq \beta(z, \xi) \int |w_n(z') - w(z')| Q(z, \diff z')
        \leq \beta(z, \xi) \|w_n - w\|_\infty.
    \end{equation*}
    Integrating over $\xi$ and taking the supremum over $z$ gives
    \begin{equation*}
        \|D w_n - D w\|_{\ell\Xsf}
        = \sup_{z \in \Zsf} \int |D w_n - D w|(z,\xi) \, \phi(\diff \xi)
        \leq \hat\beta_{\sup} \, \|w_n - w\|_\infty,
    \end{equation*}
    where $\hat\beta_{\sup} \coloneq \sup_{z \in \Zsf} \int \beta(z,\xi) \, \phi(\diff\xi)$
    is finite by \cref{a:fv}. Hence $D$ is continuous from $(b\Zsf, \|\cdot\|_\infty)$
    to $(\ell\Xsf, \|\cdot\|_{\ell\Xsf})$.
\end{proof}

\begin{proof}[Proof of Proposition~\ref{p:fk}]
    Lemma~\ref{l:stf} tells us that $(\ell \Xsf, S)$ and $(b\Zsf, T)$ are 
    strongly semiconjugate under $E,D$.
    We showed in Lemma~\ref{l:gs} that $(b\Zsf, T)$ is globally stable 
    with unique fixed point in $b\Zsf$. 
    We showed in Lemma~\ref{l:gcont} that $D$ is continuous.
    All the claims in Proposition~\ref{p:fk} now follow from \cref{t:sosft}.
\end{proof}

\subsection{An Analytical Illustration} \label{ss:analytical}

The result above, Proposition~\ref{p:fk}, holds under the mild
eventual-contraction condition of Assumption~\ref{a:fv}. If we strengthen this
assumption to requiring a constant discount factor, we can also see how shifting
from $S$ to $T$ moves us to a space of bounded functions and, in consequence,
provides a transparent contraction result. This helps to illustrate how the
encoding-decoding methodology can yield analytical convenience as well as
computational efficiency.  To provide this demonstration, we assume that the
discount factor is constant: $\beta(z, \xi) \equiv \beta$ for some constant
$\beta < 1$.

To begin this discussion, notice that analyzing the fixed point problem via the original
operator $S$ is complicated by the fact that $S$ acts on the space $\ell\Xsf$ of potentially
unbounded functions. In contrast, $T$ acts on 
the space of bounded functions $b\Zsf$, and, because of this, naturally becomes
a contraction mapping under the supremum norm
$\|\cdot\|_\infty$, with modulus $\beta$. Indeed, for any $w_1, w_2 \in b\Zsf$
and $z \in \Zsf$, an application of the elementary bound $|\max\{p, a\} -
\max\{p, b\}| \leq |a-b|$ yields
\begin{align*}
    |(T w_1)(z) - (T w_2)(z)|
    &\leq
     \left|\int \int \beta(z, \xi) \bigl(w_1(z') - w_2(z')\bigr) Q(z, \diff z')
    \phi(\diff \xi) \right| \\
    &\leq \int \beta(z, \xi) \phi(\diff \xi)
    \int |w_1(z') - w_2(z')| Q(z, \diff z') \\
    &\leq \beta \|w_1 - w_2\|_\infty.
\end{align*}
Taking the supremum on the left-hand side yields
$\|T w_1 - T w_2\|_\infty \leq \beta \|w_1 - w_2\|_\infty$.  By the
Banach fixed point theorem, $T$ has a unique fixed point $\bar w \in b\Zsf$, and
the iterates $T^n w$ converge to $\bar w$ for any $w \in b\Zsf$.  Hence, in this special case the Banach fixed point theorem
alone establishes existence, uniqueness, and global convergence of the fixed
point on the ``nicer'' space $b\Zsf$, and Theorem~\ref{t:sosft} transfers these
properties back to the original problem on $\ell\Xsf$, recovering
Proposition~\ref{p:fk} by a more elementary route under the stronger assumption.

\section{Further Applications} \label{s:fa}

In this section we provide three further applications. The first is a turn-based
zero-sum stochastic game, the second is a standard
asset pricing problem in an economy with $K$ assets,
 and the third concerns asset
pricing with heterogeneous beliefs.  In that last application, we show how the types of transformations
discussed above can also be used to eliminate bias for a solution method based on
stochastic approximation.

\subsection{A Game-theoretic Application}\label{ss:game}

Consider a finite, discounted, alternating turn-based zero-sum stochastic game
(see, e.g., \cite{shapley1953stochastic}, \cite{mertens2015repeated}, or
\cite{solan2015stochastic}).
The state space $\Xsf_1$ is all states in which player 1 chooses an action,
and the state space $\Xsf_2$ is all states in which player 2 chooses an action.
Suppose that $\Xsf_1$ and $\Xsf_2$ are disjoint. Play alternates between the two: every
transition out of $\Xsf_1$ lands in $\Xsf_2$ and every transition out of $\Xsf_2$ lands in
$\Xsf_1$, so the game is bipartite in structure.

Let $\beta \in (0, 1)$ be the discount factor. 
If the current state is $s \in \Xsf_1$, player 1 chooses an action $a \in A$. 
Player 1 receives current payoff $r_1(s, a)$ and player 2 receives the negative of that payoff. 
The next state is drawn according to a transition probability $P_1(s, a, \cdot)$ over $\Xsf_2$.
If the current state is $s \in \Xsf_2$, player 2 chooses an action $b \in B$. 
Player 1 receives current payoff $r_2(s, b)$, and player 2 receives the negative of that payoff.
The next state is drawn according to a transition probability $P_2(s, b, \cdot)$ over $\Xsf_1$.

Let $v_1$ be the value function on $\Xsf_1$, and let $v_2$ be the value function on $\Xsf_2$. 
The value $v_1(s)$ is the equilibrium continuation payoff to player 1
when the current state is $s \in \Xsf_1$; since play then moves to $\Xsf_2$, the
continuation is evaluated with $v_2$.
Symmetrically, $v_2(s)$ is the equilibrium continuation payoff to player 1
when the current state is $s \in \Xsf_2$, and its continuation is evaluated with $v_1$.
For simplicity, we assume that $\Xsf_1$, $\Xsf_2$, $A$, and $B$ are finite.

The Bellman equation is
\begin{equation*}
    v_1(s) = \max_{a \in A}   
    \left\{r_1(s, a) + \beta \sum_{s' \in \Xsf_2} v_2(s') P_1(s, a, s')\right\}, 
    \quad (s \in \Xsf_1),
\end{equation*}
\begin{equation*}
    v_2(s) = \min_{b \in B} 
    \left\{r_2(s, b) + \beta \sum_{s' \in \Xsf_1} v_1(s') P_2(s, b, s')\right\},
    \quad (s \in \Xsf_2).
\end{equation*}
The minimization appears because values are written from player 1's point of view,
while player 2 chooses $b$.

To solve the equilibrium values $v_1$ and $v_2$, we 
define the operators $E \colon \RR^{\Xsf_2} \to \RR^{\Xsf_1}$ 
and $D \colon \RR^{\Xsf_1} \to \RR^{\Xsf_2}$ by
\begin{equation}
    (E v_2)(s) = \max_{a \in A}   
    \left\{r_1(s, a) + \beta \sum_{s' \in \Xsf_2} v_2(s') P_1(s, a, s')\right\}, 
    \quad (s \in \Xsf_1).
\end{equation}
\begin{equation}
    (D v_1)(s) = \min_{b \in B} \left\{r_2(s, b) + \beta \sum_{s' \in \Xsf_1}
    v_1(s') P_2(s, b, s')\right\},
    \quad (s \in \Xsf_2).
\end{equation}

Define the operators $T = E \circ D$ and $S = D \circ E$, and let $V_2 = \RR^{\Xsf_2}$
and $V_1 = \RR^{\Xsf_1}$, each equipped with the supremum norm $\|\cdot\|_\infty$. Since
$S = D \circ E$ on $V_2$ and $T = E \circ D$ on $V_1$, the systems $(V_2, S)$ and
$(V_1, T)$ are strongly semiconjugate under $E, D$. Both $E$ and $D$ are
$\beta$-Lipschitz under the supremum norm. Indeed,
for any $w_1, w_2 \in V_2$ and $s \in \Xsf_1$, 
$|\max_a f(a) - \max_a g(a)| \leq \max_a |f(a) - g(a)|$ yields
\begin{align*}
    |(E w_1)(s) - (E w_2)(s)|
    &\leq \beta \max_{a \in A} \left| \sum_{s' \in \Xsf_2}
    \bigl(w_1(s') - w_2(s')\bigr) P_1(s, a, s') \right| \\
    &\leq \beta \|w_1 - w_2\|_\infty,
\end{align*}
since each $P_1(s, a, \cdot)$ is a probability distribution. Taking the supremum
over $s \in \Xsf_1$ gives $\|E w_1 - E w_2\|_\infty \leq \beta \|w_1 - w_2\|_\infty$,
and the same argument, applied to the minimum, shows that $D$ is $\beta$-Lipschitz.
Hence $T = E \circ D$ is a $\beta^2$-contraction on the complete space $V_1$, so, by
the Banach fixed point theorem, $(V_1, T)$ is globally stable with a unique fixed
point $\bar v_1$, and $T^n w \to \bar v_1$ for every $w \in V_1$. Since $D$ is also
continuous, Theorem~\ref{t:sosft}---with $(V_1, T)$ and $(V_2, S)$ in the roles of
$(W, T)$ and $(V, S)$---transfers global stability to the higher-dimensional
system: $(V_2, S)$ has a unique fixed point $\bar v_2 \coloneq D \bar v_1$, and
\begin{equation*}
    D T^n w \to \bar v_2
    \quad \text{as} \quad n \to \infty,
    \qquad \text{for every } w \in V_1.
\end{equation*}
The pair $(\bar v_1, \bar v_2)$ solves the coupled Bellman equations above. We may
therefore compute the equilibrium by iterating $T$ on the lower-dimensional space
$V_1$ and decoding once via $D$, without ever iterating the higher-dimensional
operator $S$.

When $|\Xsf_1|$ is smaller than $|\Xsf_2|$, iterating with $T$ rather than $S$ reduces the dimension of
the value vector, and also reduces per-iteration time under a
sparse-reachability condition. The memory gain from the dimensionality reduction is immediate, since the iterate of
$T$ lives in $\RR^{\Xsf_1}$ rather than $\RR^{\Xsf_2}$. The time gain holds when the
transition kernel $P_1$ reaches only a subset of $\Xsf_2$ from $\Xsf_1$, so that an
iteration of $T$ can be carried out without a full evaluation of $D$ on $\Xsf_2$.

To quantify both effects, let $n_1 \coloneq |\Xsf_1|$, $n_2 \coloneq |\Xsf_2|$,
$a \coloneq |A|$, and $b \coloneq |B|$, and let
\begin{equation*}
    d_1 \coloneq \max_{s \in \Xsf_1,\, u \in A}
    \left| \operatorname{supp} P_1(s, u, \cdot) \right|,
    \qquad
    d_2 \coloneq \max_{s \in \Xsf_2,\, u \in B}
    \left| \operatorname{supp} P_2(s, u, \cdot) \right|
\end{equation*}
be the maximal transition-support sizes of $P_1$ and $P_2$, so that $d_1 \leq n_2$
and $d_2 \leq n_1$. Let
\begin{equation*}
    R_2 \coloneq \bigcup_{s \in \Xsf_1,\, u \in A}
    \operatorname{supp} P_1(s, u, \cdot) \subseteq \Xsf_2,
    \qquad \rho_2 \coloneq |R_2|,
\end{equation*}
be the set of player-2 states reached with positive probability from some
player-1 state-action pair.

We measure cost by the number of arithmetic operations, using standard order
notation: $f = O(g)$ means that $f$ is bounded above by a constant multiple of
$g$, while $f = \Theta(g)$ means that $f$ and $g$ are of the same order, i.e.,
each is bounded above by a constant multiple of the other. In this example,
one application of $E \colon \RR^{\Xsf_2} \to \RR^{\Xsf_1}$ costs $O(n_1 a d_1)$ and one
full application of $D \colon \RR^{\Xsf_1} \to \RR^{\Xsf_2}$ costs $O(n_2 b d_2)$, so an
iteration of $S = D \circ E$ costs $O(n_1 a d_1 + n_2 b d_2)$. An iteration of
$T = E \circ D$ avoids a full decode: since $(E(Dv_1))(s)$ involves $(Dv_1)(s')$ only for states
$s'$ reached from $\Xsf_1$, the decode $Dv_1$ is needed only on $R_2$, so the iteration
costs $O(n_1 a d_1 + \rho_2 b d_2)$, followed by a single decode on
$\Xsf_2$ at cost $O(n_2 b d_2)$ when the equilibrium value $v_2$ is required. Because $R_2 \subseteq \Xsf_2$, we have
$\rho_2 \leq n_2$, so a $T$-iteration never costs more than an $S$-iteration; the saving
grows as $\rho_2$ falls below $n_2$---that is, as fewer player-2 states are reachable
from the player-1 states. Finally, since $S$ and $T$ are $\beta^2$-contractions, reaching tolerance
$\varepsilon$ takes $K_\varepsilon = O(\log(1/\varepsilon) / |\log \beta^2|)$
iterations. Hence, computing the equilibrium values to tolerance $\varepsilon$
costs $O(K_\varepsilon (n_1 a d_1 + n_2 b d_2))$ by direct iteration on $S$, and
$O(K_\varepsilon (n_1 a d_1 + \rho_2 b d_2) + n_2 b d_2)$ by reduced iteration on
$T$ followed by a final decode.

In the sparse regime where $a, b, d_1, d_2 = O(1)$, $\rho_2 = O(n_1)$, and
$n_1 \ll n_2$, direct iteration costs $O(K_\varepsilon n_2)$, while reduced
iteration costs $O(K_\varepsilon n_1 + n_2)$ when the full $\Xsf_2$-value is decoded at
the end. Thus, if we only compute the low-dimensional value $\bar v_1$ on $\Xsf_1$, the
speedup is of order $\Theta(n_2 / n_1)$, and when the full value $\bar v_2 = D \bar v_1$ is
also computed, the speedup is of order $\Theta(\min\{ K_\varepsilon,\, n_2 / n_1 \})$.%
\footnote{Since $K_\varepsilon n_1 + n_2 = \Theta(\max\{K_\varepsilon n_1,\, n_2\})$
and $X / \max\{A, B\} = \min\{X/A,\, X/B\}$, the speedup is
$K_\varepsilon n_2 / (K_\varepsilon n_1 + n_2) = \Theta(\min\{n_2/n_1,\, K_\varepsilon\})$.
When the iteration term $K_\varepsilon n_1$ dominates the final decode $n_2$, the
speedup is the full dimensional ratio $n_2/n_1$. When the decode dominates, it is
capped at $K_\varepsilon$.}
Moreover, even when
$R_2 = \Xsf_2$, where the time complexity is of the same order, the value-vector
memory reduction still holds, since the iterate of $T$ has dimension $n_1$ rather
than $n_2$.

\subsection{An Asset Pricing Problem}\label{ss:ap}

Next we treat an asset pricing problem in an economy with $K$ assets. Assuming
absence of arbitrage and homogeneous beliefs, the current price of the $i$-th
asset, which prices a claim to the dividend process $(Y^i_t)_{t \geq 0}$, obeys
\begin{equation}\label{eq:pa}
	P_t^i = Y_t^i + \EE_t M_{t+1} P_{t+1}^i, 
	\qquad t = 0, 1, \ldots,
\end{equation}
where $(M_t)_{t \geq 1}$ is a stochastic discount factor (SDF) process (see, e.g.,
\cite{kreps1981arbitrage} or \cite{hansen2009pricing}).\footnote{We are pricing
a \emph{cum-dividend} contract, so that the owner of the asset at the start of
time $t$ receives the current dividend.} We suppose that dividends and the SDF
process are both functions of a Markov state process $(Z_t)_{t \geq
0}$ that takes values in measurable space $(\Zsf, \zZ)$. In particular, $M_{t+1}
      = m(Z_t, Z_{t+1})$ for some measurable function $m \colon \Zsf \times \Zsf
      \to (0, \infty)$, while $Y_t^i = d_i(Z_t, \xi_{t,i})$.  The map $d_i$ is also
      assumed to be measurable and takes values in $\RR_+$, while $\xi_{t,i}$ is
      the $i$-th component of the {\sc iid} vector process $(\xi_t)$, which is
      independent of $(Z_t)_{t \geq 0}$ and takes values in Borel subset $\Esf$
      of $\RR^K$.  Thus, at each time $t$, the
      dividend of the $i$-th asset depends on the persistent macro-level common
      factors embedded in $(Z_t)$ and the idiosyncratic and transitory shock
      components in $(\xi_t)$.

Due to the Markov property of the primitives, we seek
a stationary Markov solution for prices, which we write as $P^i_t = v_i(Z_t, \xi_{t,i})$.
In view of \eqref{eq:pa}, the function $v_i$ will satisfy the functional
equation 
\begin{equation}\label{eq:gap3}
    v_i(z, \xi_i)
    = d_i(z, \xi_i) + \int \int m(z, z')  v_i(z', \xi_i') Q(z, \diff z') \phi(\diff \xi'),
\end{equation}
where $Q$ is a stochastic kernel generating the state process $(Z_t)$ and $\phi$
is the distribution of $\xi$. We can stack the $K$ functional equations in
\eqref{eq:gap3} into a functional equation for the cross-section of assets:
\begin{equation}\label{eq:apv}
	v(z, \xi) 
	= d(z, \xi) + \int \int m(z, z') v(z', \xi') Q(z, \diff z') \phi(\diff \xi').
\end{equation}
Here $v = (v_1, \ldots, v_K)$ and so on.

In order to avoid placing boundedness restrictions on the dividend process,
we embed this problem in a space of integrable functions.  In particular, 
we seek a fixed point of the mapping
\begin{equation*}
	(Sv)(z, \xi) = d(z, \xi) + \int \int 
	m(z, z') v(z', \xi') Q(z, \diff z') \phi(\diff \xi')
\end{equation*}
in the space $L_1(\psi \times \phi) \coloneq L_1(\Zsf \times \Esf, \RR^K, \psi \times \phi)$, the
Banach space of (equivalence classes of) measurable functions
$v \colon \Zsf \times \Esf \to \RR^K$ with
\begin{equation}\label{eq:pp}
	\| v \|_1
	\coloneq \int \int
        \sum_{i=1}^K | v_i(z, \xi) | \psi(\diff z) \phi(\diff \xi)
	< \infty.
\end{equation}
Each component $v_i$ depends only on the $i$-th coordinate $\xi_i$ of $\xi$.
Here and below, $\psi$ is the unique stationary distribution of $Q$
and $\| \cdot \|_1$ is a norm over the function space $L_1(\psi \times \phi)$.
A solution to the asset pricing problem is a fixed point for $S$ on $L_1(\psi \times \phi)$.

Similarly, we define $L_1(\psi) \coloneq L_1(\Zsf, \RR^K, \psi)$, the space of
(equivalence classes of) measurable functions $w \colon \Zsf \to \RR^K$ with
$\| w \|_{L_1(\psi)} \coloneq \int \sum_{i=1}^K |w_i(z)| \, \psi(\diff z) < \infty$.
The operator $H$ below acts componentwise, so its scalar properties extend to
$L_1(\psi)$ coordinatewise.

A lower-dimensional representation of this fixed point problem can be obtained
by working instead with the equation
\begin{equation}\label{eq:apw}
	w(z)
	= \int \int m(z, z') d(z', \xi') Q(z, \diff z') \phi(\diff \xi')
	+ \int m(z, z') w(z') Q(z, \diff z').
\end{equation}
In comparing the functional equations \eqref{eq:apw} and \eqref{eq:apv}, we see
that the function $w$ lives on the lower-dimensional space $\Zsf$, as compared
to the space $\Zsf \times \Esf$.   The extent of dimensionality reduction
depends on the number of assets $K$ and the dimensionality of each
$\xi_i$. Let the operator $T$ be defined by 
\begin{equation*}
    (Tw)(z)
	= \int \int m(z, z') d(z', \xi') Q(z, \diff z') \phi(\diff \xi')
	+ \int m(z, z') w(z') Q(z, \diff z').
\end{equation*}
We next define $E \colon L_1(\psi \times \phi) \rightarrow L_1(\psi)$ and
$D \colon L_1(\psi)\rightarrow L_1(\psi \times \phi)$ via
\begin{align*}
    & (Ev)(z) =\int \int 
        m(z, z') v(z', \xi') Q(z, \diff z') \phi(\diff \xi')
        \qquad \text{and}
        \\
    & (Dw)(z,\xi) =d(z,\xi)+w(z).
\end{align*}

\begin{assumption}\label{a:app}
    The following conditions hold:
    \begin{enumerate}
	\item The dividend function $d$ is in $L_1(\psi \times \phi)$.
	\item The stochastic kernel $Q$ is irreducible and has a unique
        stationary distribution $\psi$.
	\item The componentwise discount operator $H$ given by
        \begin{equation*}
            (H w)(z) := \int m(z, z') w(z') Q(z, \diff z') 
            \qquad (w \in L_1(\psi))
        \end{equation*}
        is eventually compact as a self-map on $L_1(\psi)$.
	\item The SDF process satisfies, with $Z_0 \sim \psi$,
	\begin{equation*}
            \lim_{n\rightarrow\infty}
	    \frac{1}{n}
	    \ln \left\{\EE_\psi\prod_{t=1}^{n}M_t \right\}<0.
	\end{equation*}
    \end{enumerate}
\end{assumption}

Under \cref{a:app} we can state the following result.

\begin{proposition}\label{p:lap}
    If \cref{a:app} holds, then the asset pricing problem \eqref{eq:apv} has a
    unique solution $\bar v$ in $L_1(\psi \times \phi)$. Moreover, for $w \in L_1(\psi)$, 
    we have $D T^n w \to \bar v$ as $n \to \infty$.
\end{proposition}

To prove Proposition~\ref{p:lap}, we begin with the following lemma:

\begin{lemma}\label{l:as_def}
    If \cref{a:app} holds, then the systems $(L_1(\psi \times \phi), S)$ 
    and $(L_1(\psi), T)$ are strongly semiconjugate under $E$ and $D$.
\end{lemma}

\begin{proof}
    We first verify that the maps $E$ and $D$ are well-defined on the
    stated spaces. For $D$, fix $w \in L_1(\psi)$. Since
    $(Dw)(z,\xi)=d(z,\xi)+w(z)$,
    \begin{align*}
        \|Dw\|_1
        &= \int \int \sum_{i=1}^K |d_i(z,\xi)+w_i(z)|
            \,\psi(\diff z)\,\phi(\diff\xi) \\
        &\leq \|d\|_1+\|w\|_{L_1(\psi)}<\infty.
    \end{align*}
    Hence $D \colon L_1(\psi)\to L_1(\psi\times\phi)$.

    For $E$, fix $v\in L_1(\psi\times\phi)$ and define
    \begin{equation*}
        \tilde v_i(z')\coloneq
        \int |v_i(z',\xi')|\,\phi(\diff\xi'),
        \qquad i=1,\ldots,K.
    \end{equation*}
    Tonelli's theorem gives
    $\|\tilde v\|_{L_1(\psi)}=\|v\|_1$. Applying the triangle inequality
    componentwise and using Tonelli again,
    \begin{align*}
        \|Ev\|_{L_1(\psi)}
        &= \int \sum_{i=1}^K |(Ev)_i(z)|\,\psi(\diff z) \\
        &\leq \int\int m(z,z')\sum_{i=1}^K\tilde v_i(z')
            \,Q(z,\diff z')\,\psi(\diff z) \\
        &= \int\sum_{i=1}^K(H\tilde v)_i(z)\,\psi(\diff z).
    \end{align*}
    By \cref{a:app}(iii), $H$ is a bounded linear operator and some power of
    $H$ is compact. Since $H$ is positive,
    \begin{align*}
        \int\sum_{i=1}^K(H\tilde v)_i(z)\,\psi(\diff z)
        &=\|H\tilde v\|_{L_1(\psi)} \\
        &\leq\|H\|\,\|\tilde v\|_{L_1(\psi)}
        =\|H\|\,\|v\|_1<\infty.
    \end{align*}
    Hence $E \colon L_1(\psi\times\phi)\to L_1(\psi)$.

    Moreover, we have $S = D \circ E$ on $L_1(\psi \times \phi)$ and $T = E \circ D$ on
    $L_1(\psi)$. Indeed, fixing $v \in L_1(\psi \times \phi)$, we have
    \begin{align*}
        (D E v)(z, \xi)
        & = d(z, \xi) + (E v)(z) \\
        & = d(z, \xi) + \int \int
	m(z, z') v(z', \xi') Q(z, \diff z') \phi(\diff \xi') \\
        & = (S v)(z, \xi).
    \end{align*}
    Fixing $w \in L_1(\psi)$, we have
    \begin{align*}
        (E D w)(z)
        & = \int \int 
            m(z, z') (D w)(z', \xi') Q(z, \diff z') \phi(\diff \xi') \\
        & = \int \int m(z, z') d(z', \xi') Q(z, \diff z') \phi(\diff \xi')
            + \int m(z, z') w(z') Q(z, \diff z') 
            \\
        & = (T w)(z).
        \qedhere
    \end{align*}
\end{proof}

\begin{lemma}\label{l:as1}
	If \cref{a:app} holds, then $T$ is globally stable.
\end{lemma}

\begin{proof}
    \cref{a:app} and Theorem~C.7 of \cite{borovivcka2021stability} imply that
    $\rho(H)$, the spectral radius of $H$, is strictly less than one. Moreover,
    $T$ has the form $Tw = Ed + Hw$. It follows that
    $T$ is eventually contracting and therefore globally stable. 
\end{proof}

\begin{proof}[Proof of Proposition~\ref{p:lap}]
    Lemma~\ref{l:as_def} tells us that $(L_1(\psi \times \phi), S)$ 
    and $(L_1(\psi), T)$ are strongly semiconjugate under $E$ and $D$.
    We showed in Lemma~\ref{l:as1} that $(L_1(\psi), T)$ is globally stable.
    Clearly, $Dw = d + w$ is continuous.  All the claims now follow from
    \cref{t:sosft}.
\end{proof}

\subsection{Asset Pricing with Heterogeneous Beliefs} \label{ss:hb}

In previous applications, we demonstrated how the encoding-decoding methodology
can enable dimensional reduction, transforming complex problems into
simpler, lower-dimensional systems. We now turn to a different 
advantage: enabling proper application of stochastic approximation (SA) methods to
solve forward-looking stochastic models.

Stochastic approximation is a fundamental technique in economics, finance, and
operations research, with deep connections to reinforcement learning and other
machine learning algorithms (see, e.g., \cite{qu2020finite},
\cite{forneron2024estimation}, and \cite{borkar2025ode}). While SA offers
powerful tools for solving high-dimensional problems with imperfect information,
its direct application to nonlinear models can introduce systematic biases and
high variance. We show how strong semiconjugacy provides a transformation that
yields an unbiased, low-variance SA formulation.

In doing so, we consider an asset pricing model with
heterogeneous beliefs (see, e.g., \cite{harrison1978speculative},
\cite{morris1996speculative}, \cite{scheinkman2003overconfidence}, and
\cite{nutz2020shorting}). The model features $N$ types of agents with differing
beliefs about the future state of the economy. Agent $i$ believes the state
process $(X_t)$ on $\Xsf$ follows the kernel $P_i$, where $i
\in I \coloneq \{1, 2, \dots, N\}$. The asset price satisfies the recursion
\begin{equation}\label{eq:poa}
    \Pi_t = \max_{i \in I} \beta \EE_i [\Pi_{t+1} + Y_{t+1}],
    \qquad t = 0, 1, \ldots
\end{equation}
where $\EE_i$ is the $i$-th agent's
expectation operator (see, e.g., \cite{harrison1978speculative}). 
Let $Y_t=d(X_t)$ denote the dividend from holding the asset.
All agents have the same discount factor $\beta \in (0, 1)$.

Seeking to solve \eqref{eq:poa} leads us to consider the functional
equation
\begin{equation} \label{eq:eq_as}
    v(x) = \max_{i \in I} \beta \int \left\{
    v(x') + d(x')
    \right\} P_i(x, \diff x')
    \qquad (x \in \Xsf),
\end{equation}
where $v(x)$ is the asset price when the state is $x$. Here, we use the notation 
$v(x)$ for consistency with the notation in the previous sections.
A solution $v$ of \eqref{eq:eq_as} is a fixed point of 
the operator $S \colon b\Xsf \to b\Xsf$ given by
\begin{equation*}
    (S v)(x) = \max_{i \in I} 
    \beta \int \left\{
    v(x') + d(x')
    \right\} P_i(x, \diff x').
\end{equation*}
In particular, $v$ is a fixed point of $S$ in $b\Xsf$ if and only if $v$ solves
the equilibrium price condition \eqref{eq:eq_as}. Throughout this subsection we
assume that $d \in b\Xsf$.

From a theoretical perspective, the operator $S$ is relatively well-behaved.
However, it is challenging to compute a fixed point of $S$ when $\Xsf$ is
high-dimensional. Even though the integral associated with each $P_i$ can be
computed via Monte Carlo, evaluation of $Sv$ couples two numerically awkward steps: for every state 
$x \in \Xsf$, we need to (i) evaluate a separate high-dimensional expectation for each index 
$i \in I$, and then (ii) take the maximum of those noisy Monte Carlo estimates. 
This ``max-of-sample-means'' construction introduces a systematic upward bias 
(see, e.g., \cite{hasselt2010double} and \cite{van2016deep}).
We discuss this issue in more detail in Appendix~\ref{s:mc_approx}.
In addition, direct Monte Carlo fails to offer the bootstrapping associated with
stochastic approximation methods and hence has relatively high variance (see,
e.g., \cite{szepesvari2022algorithms}).

To enable proper application of SA, we transform the problem using strong semiconjugacy, 
yielding an unbiased formulation amenable to low-variance stochastic approximation. 
Let $W$ be the set of bounded Borel measurable functions on $\Xsf \times I$,
equipped with the supremum norm.
Define $T \colon W \to W$ by
\begin{equation*}
    (T w)(x, i) = \beta \int \left\{
    \max_{j \in I} w(x', j) + d(x')
    \right\} P_i(x, \diff x').
\end{equation*}
Before proceeding to numerical issues and computation of the fixed point of $T$
using stochastic approximation methods, we first show that we can solve the
fixed point problem~\eqref{eq:eq_as} by solving for the fixed point of $T$.

To solve \eqref{eq:eq_as} we introduce maps $E$ and $D$ defined by
\begin{equation*}
    (E v)(x, i) = \beta \int \left\{
    v(x') + d(x')
    \right\} P_i(x, \diff x')
    \quad \text{and} \quad
    (D w)(x) = \max_{i \in I} w(x, i).
\end{equation*}

It is easy to show that $E$ and $D$ obey $E \colon b\Xsf \to W$ and $D \colon W \to b\Xsf$.
We use $E$ and $D$ to connect $S$ and $T$ as follows:

\begin{lemma}\label{l:stf_asset}
    The systems $(b\Xsf, S)$ and $(W, T)$ are strongly semiconjugate under $E,D$.
\end{lemma}

\begin{proof}
    We have $S = D \circ E$ on $b\Xsf$ and $T = E \circ D$ on $W$.
    Indeed, fixing $v \in b\Xsf$, we have
    \begin{align*}
        (D E v)(x)
        & = \max_{i \in I} (E v)(x, i) 
        \\
        & = \max_{i \in I} \beta \int 
            \left\{
                v(x') + d(x')
            \right\} P_i(x, \diff x') 
        = (S v)(x).
    \end{align*}
    
    Moreover, for each $w \in W$, we have
    \begin{align*}
        (E D w)(x, i)
        & = \beta \int 
            \left\{
                (D w)(x') + d(x')
            \right\} P_i(x, \diff x') 
        \\
        & = \beta \int 
            \left\{
                \max_{j \in I} w(x', j) + d(x')
            \right\} P_i(x, \diff x') 
        = (T w)(x, i).
        \qedhere
    \end{align*}
\end{proof}

We can now state our main result on asset pricing with heterogeneous beliefs.

\begin{proposition}\label{p:heterogeneous}
    If $\beta \in (0,1)$ and $d \in b\Xsf$, then the asset pricing problem \eqref{eq:eq_as} has a
    unique solution $\bar v$ in $b\Xsf$. Moreover, for $w \in W$, 
    \begin{equation*}
        D T^n w \to \bar v \quad \text{as} \quad n \to \infty.
    \end{equation*}
\end{proposition}

The existence and uniqueness claim in Proposition~\ref{p:heterogeneous} is established by (a)
proving that the system $(W, T)$ has a unique fixed point via the contraction
property and (b) translating this unique fixed point in the ``nice'' system $(W,
T)$ to a unique fixed point for $S$ via the strong semiconjugacy described in
Lemma~\ref{l:stf_asset}. Then we use continuity of $D$ and strong semiconjugacy to
establish the convergence results. 

\begin{lemma}\label{l:torder}
    The system $(W, T)$ is globally stable.
\end{lemma}

\begin{proof}
    For any $w_1, w_2 \in W$ and all $(x,i) \in \Xsf \times I$,
    \begin{equation*}
    |(Tw_1)(x,i)-(Tw_2)(x,i)|
    \leq \beta \int \max_j|w_1(x',j)-w_2(x',j)|P_i(x,\diff x') 
    \leq \beta \|w_1-w_2\|.
    \end{equation*}
    Taking the supremum over all $(x,i) \in \Xsf \times I$, we obtain $\|Tw_1 -
    Tw_2\| \leq \beta\|w_1 - w_2\|$. Since $\beta \in (0,1)$, $T$ is a
    contraction mapping on the complete space $W$. Hence, by the Banach fixed
    point theorem, $(W,T)$ is globally stable.
\end{proof}

\begin{proof}[Proof of Proposition~\ref{p:heterogeneous}]
    Lemma~\ref{l:stf_asset} tells us that $(b\Xsf, S)$ and $(W, T)$ are strongly
    semiconjugate under $E,D$. In Lemma~\ref{l:torder}, we proved that $(W, T)$ is
    globally stable with unique fixed point $\bar w \in W$. It is trivial that $D$ 
    is continuous. All the claims in
    Proposition~\ref{p:heterogeneous} now follow from Theorem~\ref{t:sosft}.
\end{proof}

Proposition~\ref{p:heterogeneous} shows that the fixed point problem 
\eqref{eq:eq_as} can be solved by iterating the operator $T$ on $W$ and then
mapping the fixed point back with $\bar v = D \bar{w}$. To demonstrate this with
a numerical example, and to exploit the relative benefits of the operator $T$
discussed above, we first introduce a stochastic approximation alternative to
$T$ based on noisy observations, given by
\begin{equation} \label{eq:T_hat_ash}
    (\hat T w)(x,i) = \beta\bigl[
        \max_{j \in I} w(x',j) + d(x')
    \bigr] \quad \text{where} \quad x' \sim P_i(x,\cdot).
\end{equation}
Equation~\eqref{eq:T_hat_ash} defines $\hat T w$ pointwise. To obtain
a random operator $\hat T_t \colon W \to W$ in the sense
of~\cite{mou2022optimal}, at each iteration $t$ we draw independent
samples $\{X'_{t,x,i}\}_{(x,i) \in \Xsf \times I}$ with
$X'_{t,x,i} \sim P_i(x,\cdot)$, independently across $t$, $x$, and $i$,
and evaluate \eqref{eq:T_hat_ash} at every $(x,i)$.
In the computational implementation, $\Xsf$ is a finite grid, so only finitely
many draws are needed per iteration and $W$ is finite-dimensional.

Rather than using direct Robbins--Monro stochastic approximation,
we use the ROOT-SA algorithm~\citep{mou2022optimal} to solve the asset pricing problem.
We adopt the ROOT-SA algorithm because it has better convergence and variance
properties than the Robbins--Monro algorithm in our setting. The details of the
algorithm are given in Appendix~\ref{s:rootsa}.

Figure~\ref{fig:rootsa_pi} compares the ROOT-SA approximation
$D w_M$ to the true equilibrium price function $\bar v$, with $N=2$ agents. Each agent believes 
that the state evolves according to an AR(1) process for 
$z_t = \log x_t$. Hence, for agent $i$:
\begin{equation*}
   z_{t+1} = \rho_i z_t + \nu_i \varepsilon_{t+1}
\end{equation*}
where $\varepsilon_{t+1} \sim N(0,1)$ is an {\sc iid} standard normal random variable.
We discretize the state space using the Tauchen method~\citep{tauchen1986finite} 
with $\rho_1 = 0.1$, $\nu_1 = 0.015$, $\rho_2 = 0.1$, $\nu_2 = 0.01$.

Figure~\ref{fig:rootsa_pi} shows that the ROOT-SA approximation
$D w_M$ closely matches the true
equilibrium price function $\bar v$ calculated by iterating the operator $S$ over a
grid of state values, providing numerical evidence that the ROOT-SA scheme
applied to $\hat T$ converges to the correct fixed point.

\section{Extension: Order-Preserving Transformations}\label{s:opts}

Many applications in economics and finance involve partially ordered spaces
where equilibrium objects exhibit natural monotonicity properties
(see, e.g., \cite{milgrom1994monotone} and \cite{topkis1998supermodularity}).
In this section, we extend our framework to exploit such structure.
Section~\ref{ss:pm} introduces parametric monotonicity results showing
that monotone shifts in a low-dimensional system translate into monotone shifts
in the corresponding high-dimensional system.
Section~\ref{ss:app_opts} applies these results to the real option problem
from Section~\ref{s:egf}, demonstrating that an increase in the discount factor
raises the value of the firm.

\subsection{Parametric Monotonicity}\label{ss:pm}

A \navy{partially ordered space} is a poset $(V, \preceq)$
endowed with a Hausdorff topology $\tau$ and having the property that $\preceq$
is closed with respect to $\tau$.  Here ``closed'' means that the graph of
$\preceq$ is closed in $V \times V$ under the product topology. Throughout this section, let $V$ and $W$ be two partially ordered spaces. In this setting a
map $F \colon V \to W$  is called
\navy{order preserving} if $v \preceq w$ implies $Fv \preceq Fw$.
We say that a map $F \colon V \to W$ \navy{dominates}
a map $G \colon V \to W$ if $G v \preceq F v$ for all $v
\in V$.

\begin{theorem}\label{t:monotone}
    Let $\Theta$ be a parameter space. For each $\theta \in \Theta$, suppose that
    \begin{enumerate}
        \item $(V, S_\theta)$ and $(W, T_\theta)$ are strongly semiconjugate under
        $E_\theta, D_\theta$,
        \item $E_\theta$ and $D_\theta$ are order preserving, and
        \item $(W, T_\theta)$ is globally stable with unique fixed point $w_\theta$.
    \end{enumerate}
    Then $v_\theta \coloneq D_\theta w_\theta$ is the unique fixed point of
    $(V, S_\theta)$. Moreover, if
    $E_{\theta_2}$ dominates $E_{\theta_1}$, and $D_{\theta_2}$ dominates
    $D_{\theta_1}$, where $\theta_1, \theta_2 \in \Theta$, then
    \begin{equation*}
        w_{\theta_1} \preceq w_{\theta_2},
        \quad \text{and} \quad
        v_{\theta_1} \preceq v_{\theta_2}.
    \end{equation*}
\end{theorem}

\begin{proof}[Proof of Theorem~\ref{t:monotone}]
    For any $w \in W$, we have
    \begin{equation*}
        T_{\theta_1} w = E_{\theta_1} (D_{\theta_1} w)
        \preceq E_{\theta_2} (D_{\theta_1} w)
        \preceq E_{\theta_2} (D_{\theta_2} w)
        = T_{\theta_2} w.
    \end{equation*}
    The first inequality follows from that $E_{\theta_2}$ dominates $E_{\theta_1}$,
    and the second inequality follows from monotonicity of $E_{\theta_2}$ and that
    $D_{\theta_2}$ dominates $D_{\theta_1}$. Hence, $T_{\theta_2}$ dominates
    $T_{\theta_1}$. Since $T_{\theta_2}$ is order preserving, setting $w =
    w_{\theta_1}$ and iterating gives
    \begin{equation*}
        w_{\theta_1} = T_{\theta_1}^n w_{\theta_1}
        \preceq T_{\theta_2}^n w_{\theta_1}
        \quad \text{for all } n \in \NN.
    \end{equation*}
    Since $T_{\theta_2}^n w_{\theta_1} \to w_{\theta_2}$ and the partial order
    on $W$ is closed, we conclude $w_{\theta_1} \preceq w_{\theta_2}$.
    Finally,
    \begin{equation*}
        v_{\theta_1} = D_{\theta_1} w_{\theta_1}
        \preceq D_{\theta_2} w_{\theta_1}
        \preceq D_{\theta_2} w_{\theta_2}
        = v_{\theta_2},
    \end{equation*}
    where the first inequality uses that $D_{\theta_2}$ dominates $D_{\theta_1}$,
    and the second uses monotonicity of $D_{\theta_2}$.
\end{proof}

\subsection{Applications}\label{ss:app_opts}

We now apply \cref{t:monotone} to the real option problem from
Section~\ref{s:egf}. Consider two parameterizations of the discount factor
$\beta_1$ and $\beta_2$ with $\beta_1(x) \leq \beta_2(x)$ for all $x \in
\Xsf$.  Each parameterization generates its own operator $S_k$ on $\ell \Xsf$
and $T_k$ on $b\Zsf$ via the definitions in Section~\ref{ss:fv}, along with
encoder $E_k$ and decoder $D_k$ for $k \in \{1,2\}$.  In particular, 
we recall that $D_k$ and $E_k$ are defined via
\begin{equation*}
    (D_k w)(z, \xi) = \max
    \left\{
        \pi(z, \xi), \; -c(z, \xi) + \beta_k(z, \xi) \, \int w(z') Q(z, \diff z')
    \right\},
\end{equation*}
and
\begin{equation*}
    (E_k v)(z) = (E v)(z) = \int v(z, \xi) \phi(\diff \xi).
\end{equation*}

We equip $\ell \Xsf$ and $b\Zsf$ with their respective pointwise orders
(for $\ell\Xsf$, we define $v \geq v'$ whenever $v(z, \xi) \geq v'(z, \xi)$
for $\phi$-almost every $\xi$ and every $z$).
Let $b_+\Zsf \coloneq \{w \in b\Zsf : w \geq 0\}$ and
$\ell_+\Xsf \coloneq \{v \in \ell\Xsf : v \geq 0\}$ denote the nonnegative
cones. The maps $E_k$ and $D_k$ are order preserving since the integral and the
max preserve the pointwise order. The encoder $E_k = E$ is independent of
$\beta_k$, so $E_2$ trivially dominates $E_1$. The semiconjugacy
$S_k = D_k \circ E_k$, $T_k = E_k \circ D_k$ on the restricted cones follows
from the same calculation as in Lemma~\ref{l:stf}. The following lemma 
verifies conditions for \cref{t:monotone}.

\begin{lemma}\label{l:nonneg_cone}
    Suppose $\pi \geq 0$ pointwise and \cref{a:fv} holds for each
    $\beta_k$ ($k \in \{1,2\}$). Then
    \begin{enumerate}
        \item $T_k$ is a self-map on $b_+\Zsf$, and the unique fixed point
            $\bar w_k$ of $T_k$ in $b\Zsf$ lies in $b_+\Zsf$,
        \item $D_2$ dominates $D_1$ on $b_+\Zsf$.
    \end{enumerate}
\end{lemma}

\begin{proof}
    For (i), fix $w \in b_+\Zsf$. Then
    $(D_k w)(z,\xi) \geq \pi(z,\xi) \geq 0$, so
    $T_k w = E(D_k w) \geq 0$, confirming that $T_k$ maps $b_+\Zsf$ into
    itself. By \cref{l:gs}, $(b\Zsf, T_k)$ is globally stable, so
    $T_k^n w \to \bar w_k$ for any $w \in b_+\Zsf$. Since $b_+\Zsf$ is
    invariant under $T_k$, every iterate $T_k^n w$ lies in $b_+\Zsf$.
    Because $b_+\Zsf$ is closed, $\bar w_k \in b_+\Zsf$.

    For (ii), fix $w \in b_+\Zsf$. Then $\int w(z') \, Q(z, \diff z') \geq 0$,
    so $\beta_1 \leq \beta_2$ gives
    $-c + \beta_1 \int w \, \diff Q \leq -c + \beta_2 \int w \, \diff Q$.
    Taking the max with $\pi$ on both sides yields
    $(D_1 w)(z,\xi) \leq (D_2 w)(z,\xi)$ for all $(z,\xi)$.
\end{proof}

\begin{proposition}\label{p:beta_mono}
    Suppose $\pi \geq 0$ pointwise and \cref{a:fv} holds for both $\beta_1$
    and $\beta_2$. Then the unique fixed points $\bar v_k$ of $S_k$ satisfy
    $\bar v_1 \preceq \bar v_2$.
\end{proposition}

\begin{proof}
    By \cref{l:nonneg_cone}, the fixed points $\bar w_k$ lie in $b_+\Zsf$
    and $D_2$ dominates $D_1$ on $b_+\Zsf$. All hypotheses of
    \cref{t:monotone} are satisfied on $(b_+\Zsf, T_k)$ and
    $(\ell_+\Xsf, S_k)$, yielding $\bar w_1 \preceq \bar w_2$ and
    $\bar v_1 = D_1 \bar w_1 \preceq D_2 \bar w_2 = \bar v_2$.
\end{proof}

\cref{p:beta_mono} confirms the intuitive economic insight: when the firm
discounts the future less (higher $\beta$), it places greater weight on the
option to wait and enter later, thereby increasing the overall value of the
firm.

\appendix

\section{Implementation Details for Section~\ref{ss:comp_fv2}}\label{s:impl_details}

\subsection{Parameter Values and Experimental Setup} \label{a:params}

Table~\ref{tab:parameters} provides a comprehensive overview of 
all parameter values used in the computational experiments presented 
in Figure~\ref{fig:numba_speedup}. The computation was performed on a machine with
an AMD Ryzen 9 9900X and 32GB RAM. The experiment was run on a single thread 
to ensure a fair comparison between the two methods.

\begin{table}[ht]
\centering
\caption{Parameter specification for computational experiments}
\label{tab:parameters}
\begin{tabular}{lll}
\hline
\textbf{Parameter} & \textbf{Value} & \textbf{Description} \\
\hline
\multicolumn{3}{c}{Grid Specifications} \\
$n_z = (n_0, n_1, n_2, n_3)$ & $(4, 4, 4, 4)$ & Persistent state grid sizes \\
$m_\xi = (m_0, m_1, m_2, m_3)$ & $(m, m, m, m)$ & Transient shock grid sizes \\
\hline
\multicolumn{3}{c}{AR(1) Process Parameters} \\
$\rho = (\rho_0, \rho_1, \rho_2, \rho_3)$ & $(0.9, 0.9, 0.9, 0.9)$ 
& Persistence parameters \\
$\nu = (\nu_0, \nu_1, \nu_2, \nu_3)$ & $(0.1, 0.1, 0.2, 0.2)$ 
& Volatility parameters \\
\hline
\multicolumn{3}{c}{Shock Distribution Parameters} \\
$\mu = (\mu_0, \mu_1, \mu_2, \mu_3)$ & $(0.0, 0.0, 0.0, 0.0)$ 
& Log-normal means \\
$\sigma = (\sigma_0, \sigma_1, \sigma_2, \sigma_3)$ & $(0.1, 0.1, 0.1, 0.1)$ 
& Log-normal std. deviations \\
\hline
\multicolumn{3}{c}{Economic Parameters} \\
$\theta$ & 0.5 & Production elasticity parameter \\
$\beta_{\min}$ & 0.05 & Lower bound for discount factor \\
$\beta_{\max}$ & 0.95 & Upper bound for discount factor \\
\hline
\end{tabular}
\end{table}

\section{The ROOT-SA Algorithm}\label{s:rootsa}

Algorithm~\ref{alg:rootsa} adapts ROOT-SA~\citep{mou2022optimal} to the
heterogeneous-beliefs model of Section~\ref{ss:hb}. Throughout this appendix,
$\Xsf$ denotes the finite computational grid and $W$ the sup-norm space of
functions on $\Xsf\times I$.
The algorithm maintains two sequences: the parameter sequence $\{w_t\}$,
which is the main iterate, and an auxiliary variance-reduced sequence
$\{\theta_t\}$.
The key feature of ROOT-SA is that, at each step $t > B_0$, the
same stochastic oracle $\hat T_t$ is evaluated at both $w_{t-1}$
and $w_{t-2}$ to form the variance-reduced update~\citep[Algorithm~1]{mou2022optimal}.
The burn-in period $B_0$ is used to initialize $\theta$, and the
output $w_M$ approximates the fixed point $\bar{w}$ of $T$.

The stepsize $\alpha$ and burn-in length $B_0$ must satisfy the
admissibility conditions of~\cite{mou2022optimal}: specifically,
$\alpha$ must lie in the range
\begin{equation*}
    \alpha \in \biggl(0,\;
    \frac{(1-\beta)^2}
         {c\,\beta^2\, J_2^2(\Gamma,\rho_M)\,\log(M/\delta)}
    \biggr]
\end{equation*}
and $B_0 \geq c\, (1-\beta)^{-2}\alpha^{-1}\log(M/\delta)$, where $c>0$ is
universal and
$J_2$ denotes the order-$2$ Dudley entropy integral of the dual ball.
In our experiments, we set
$\alpha = a_0 / \sqrt{M - B_0}$ with $a_0 > 0$ chosen small enough
that $\alpha$ falls in the admissible range for the given $M$ and $\delta$.
Since $D$ is continuous on $(W, \|\cdot\|)$, $D w_M$ approximates the equilibrium price
function $\bar v = D\bar w$.

\begin{algorithm}[ht]
    \caption{\label{alg:rootsa}ROOT-SA for Heterogeneous-Beliefs Asset Pricing}
    \small
    \DontPrintSemicolon
    \KwIn{Finite grid $\Xsf$; initial approximation $w_0 \in W$; burn-in length $B_0 \geq 2$;
    stepsize $\alpha > 0$; total iterations $M$}
    \KwOut{Approximation $w_M$ of the fixed point $\bar w$}
    $\theta_0 \gets 0$; \quad $w_{-1} \gets w_0$ \tcp*{initialization}
    \For{$t = 1$ \KwTo $M$}{
        \For{each $(x,i) \in \Xsf \times I$}{
            Sample $X'_{t,x,i} \sim P_i(x,\cdot)$ independently\;
        }
        $\delta_t^{\mathrm{cur}} \gets \hat T_t w_{t-1} - w_{t-1}$
        \tcp*{current-iterate residual}

        \eIf{$t \leq B_0$}{
            $\theta_t \gets \theta_{t-1} +
            \dfrac{1}{B_0}\,\delta_t^{\mathrm{cur}}$
            \tcp*{accumulate burn-in average}
            $w_t \gets w_0$  \tcp*{freeze main iterate in burn-in}
        }{
            $\delta_t^{\mathrm{prev}} \gets \hat T_t w_{t-2} - w_{t-2}$
            \tcp*{same oracle at previous iterate}
            $\theta_t \gets \delta_t^{\mathrm{cur}} + \dfrac{t-1}{t}\bigl(\theta_{t-1} - \delta_t^{\mathrm{prev}}\bigr)$ \\
            $w_t \gets w_{t-1} + \alpha \theta_t$ \tcp*{SA step on $w$}
        }
    }
    \KwRet{$w_M$}
\end{algorithm}

We now verify that our setting fits the framework
of~\cite{mou2022optimal}. Their theory applies to a separable Banach space
$(\mathbb{V}, \|\cdot\|)$ equipped with a population operator $\mathbf{h}$ and
a sequence of stochastic oracles $\{\mathbf{H}_t\}_{t \geq 1}$, each
mapping $\mathbb{V}$ to itself, satisfying assumptions (A1)--(A3) of that paper.
In our setting, $\mathbb{V} = W
\cong \RR^{|\Xsf| \times N}$ is the
finite-dimensional space of functions on $\Xsf \times I$, equipped with the sup-norm,
$\mathbf{h} = T$, and the stochastic oracle $\mathbf{H}_t = \hat T_t$ is
the random operator constructed in Section~\ref{ss:hb} by drawing
independent samples for each $(x,i) \in \Xsf \times I$.
We verify the three conditions below.

\begin{proposition} \label{p:rootsa}
    The following three claims hold:
    \begin{enumerate}
    \item $T$ is a contraction mapping on $W$ with factor $\beta \in (0,1)$,
    \item for each $t$, the random operator $\hat T_t \colon W \to W$ is $\beta$-Lipschitz continuous almost surely, and
    \item for any fixed $w \in W$, the noise variables $\{\epsilon_t(w) \coloneq
    \hat T_t w - T w\}_{t \geq 1}$ are {\sc iid} $W$-valued random elements with zero mean;
    moreover, at the fixed point $w = \bar w$,
    $\|\epsilon_t(\bar w)\| \leq b_*$ almost surely, where
    $b_* \coloneq 2\beta(\|\bar w\| + M_d)$ and
    $M_d \coloneq \max_{x \in \Xsf}|d(x)|$.
    \end{enumerate}
\end{proposition}

\begin{proof}
    
    The first claim is verified by \cref{l:torder}.
    
    Next, we verify that each $\hat T_t$ is $\beta$-Lipschitz almost surely.
    For any realization of the samples $\{X'_{t,x,i}\}$ and any
    $w_1, w_2 \in W$, every $(x,i)$ satisfies
    \begin{equation*}
    |(\hat T_t w_{1})(x,i)-(\hat T_t w_{2})(x,i)|
          =\beta \bigl|\max_j w_1(X'_{t,x,i}, j)-\max_j w_2(X'_{t,x,i}, j)\bigr|
          \leq \beta \|w_1-w_2\|.
    \end{equation*}
    Taking the supremum over $(x,i)$ gives
    $\|\hat T_t w_1 - \hat T_t w_2\| \leq \beta\|w_1 - w_2\|$ almost surely.

    Fix $w\in W$. We show that
    $\{\epsilon_t(w)\}_{t \geq 1}$ is {\sc iid}, mean-zero, and bounded in $W$.
    Recall that $\epsilon_t(w) = \hat T_t w - T w \in W$.
    For each $(x,i)$,
    \begin{equation*}
    \epsilon_t(w)(x,i)
    =\beta\Bigl[\max_j w(X'_{t,x,i},j)+d(X'_{t,x,i})\Bigr]
    -\beta \int \bigl[\max_j w(x',j)+d(x')\bigr] P_i(x, \diff x').
    \end{equation*}
    Since $\EE[(\hat T_t w)(x,i)] = (Tw)(x,i)$ for each $(x,i)$,
    we have $\EE[\epsilon_t(w)] = 0$ in $W$.
    Because the families $\{X'_{t,x,i}\}_{x,i}$ are independent across $t$,
    the $\{\epsilon_t(w)\}_{t \geq 1}$ are {\sc iid} random elements of $W$.

    It remains to verify the almost-sure bound at $w = \bar w$.
    For any $(x,i)\in\Xsf\times I$ and any realization of $X'_{t,x,i}$,
    \begin{align*}
        |(\hat T_t \bar w)(x,i)|
        &= \beta\bigl|\max_j \bar w(X'_{t,x,i},j)+d(X'_{t,x,i})\bigr|
        \leq \beta(\|\bar w\| + M_d), \\
        |(T \bar w)(x, i)| &\leq \beta(\|\bar w\| + M_d).
    \end{align*}
    Hence $|\epsilon_t(\bar w)(x,i)| \leq 2 \beta(\|\bar w\| + M_d)$.
    Taking the supremum over $(x,i)$ gives
    $\|\epsilon_t(\bar w)\| \leq 2 \beta(\|\bar w\| + M_d) = b_*$ almost surely.

    Therefore, all three claims in the proposition are verified.
\end{proof}

Proposition~\ref{p:rootsa} verifies assumptions (A1)--(A3)
of~\cite{mou2022optimal}, with
contraction modulus $\gamma = \beta$,
Lipschitz constant $L = \beta$,
and fixed-point noise bound $b_*$.
Together with Algorithm~\ref{alg:rootsa} and the admissible stepsize and
burn-in conditions stated above, all hypotheses of their
Theorem~1 are satisfied.

We apply Theorem~1 of~\cite{mou2022optimal} in conjunction with
their Appendix~A restarting scheme, which removes dependence on the
initial operator defect $\|T w_0 - w_0\|$.
The restarting procedure runs a sequence of short warm-started epochs, each
of length $2cB_0$, followed by a final epoch using the remaining sample budget.
For bounded initializations on a finite grid, the initial-defect condition
required by the restart argument is automatically satisfied.
With this restarted wrapper, for any
confidence level $1-\delta\in(0,1)$
and a sample size $M$ satisfying the minimum-size requirement
\begin{equation*}
    M \geq \frac{c}{(1-\beta)^4}\,\beta^2\, J_2^2(\Gamma,\rho_M)\,
    \log^2\!\bigl(\tfrac{M}{\delta}\bigr),
\end{equation*}
the ROOT-SA iterate $w_M$ satisfies the operator-defect bound
\begin{equation}\label{eq:rootsa_defect}
    \|T w_M - w_M\| \leq
    \frac{C_1}{\sqrt{M}}
    \Bigl[\EE[\|\mathcal{W}\|] + \sqrt{\sigma_\Gamma^2(\mathcal{W})\log(1/\delta)}\Bigr]
    + \frac{C_2}{M}
    \Bigl[J_1(\Gamma,\rho_M) + \log(1/\delta)\Bigr]
\end{equation}
with probability at least $1-\delta$.
Here $\mathcal{W}$ is the zero-mean Gaussian element in $W$ with covariance matching that of
the noise $\hat T_1 \bar w - T \bar w$,
$\sigma_\Gamma^2(\mathcal{W}) \coloneq \sup_{u \in \Gamma}\EE[\langle u, \mathcal{W}\rangle^2]$
is the maximal directional variance over a skeleton set $\Gamma$ of the dual
ball, $J_q(\Gamma,\rho_M)$ denotes the order-$q$ Dudley entropy integral
(with $\rho_M$ the pseudo-metric defined in Section~3.1.1
of~\cite{mou2022optimal}),
and $C_1, C_2 > 0$ depend on $\beta$ and $b_*$.

Because $T$ is $\beta$-contractive and $\bar w = T \bar w$, the
operator-defect bound converts immediately into a fixed-point error
bound\footnote{By the triangle inequality,
$\|w_M - \bar w\| \leq \|T w_M - w_M\| + \|T w_M - T \bar w\|
\leq \|T w_M - w_M\| + \beta\|w_M - \bar w\|$; rearranging gives
$\|w_M - \bar w\| \leq (1-\beta)^{-1}\|T w_M - w_M\|$.}:
\begin{equation}\label{eq:bound_w}
    \|w_M - \bar w\| \leq (1-\beta)^{-1} \|T w_M - w_M\|.
\end{equation}
For fixed confidence level $\delta$ and admissible tuning,
the right-hand side of \eqref{eq:bound_w}
is $O(M^{-1/2})$ with probability at least $1-\delta$.
In particular, $w_M \to \bar w$ in probability as $M \to \infty$.
Since $D$ is Lipschitz on $(W, \|\cdot\|)$, we obtain
$D w_M \to D \bar w = \bar v$
in probability.

\begin{figure}[p]
    \centering
    \includegraphics[width=0.75\textwidth]{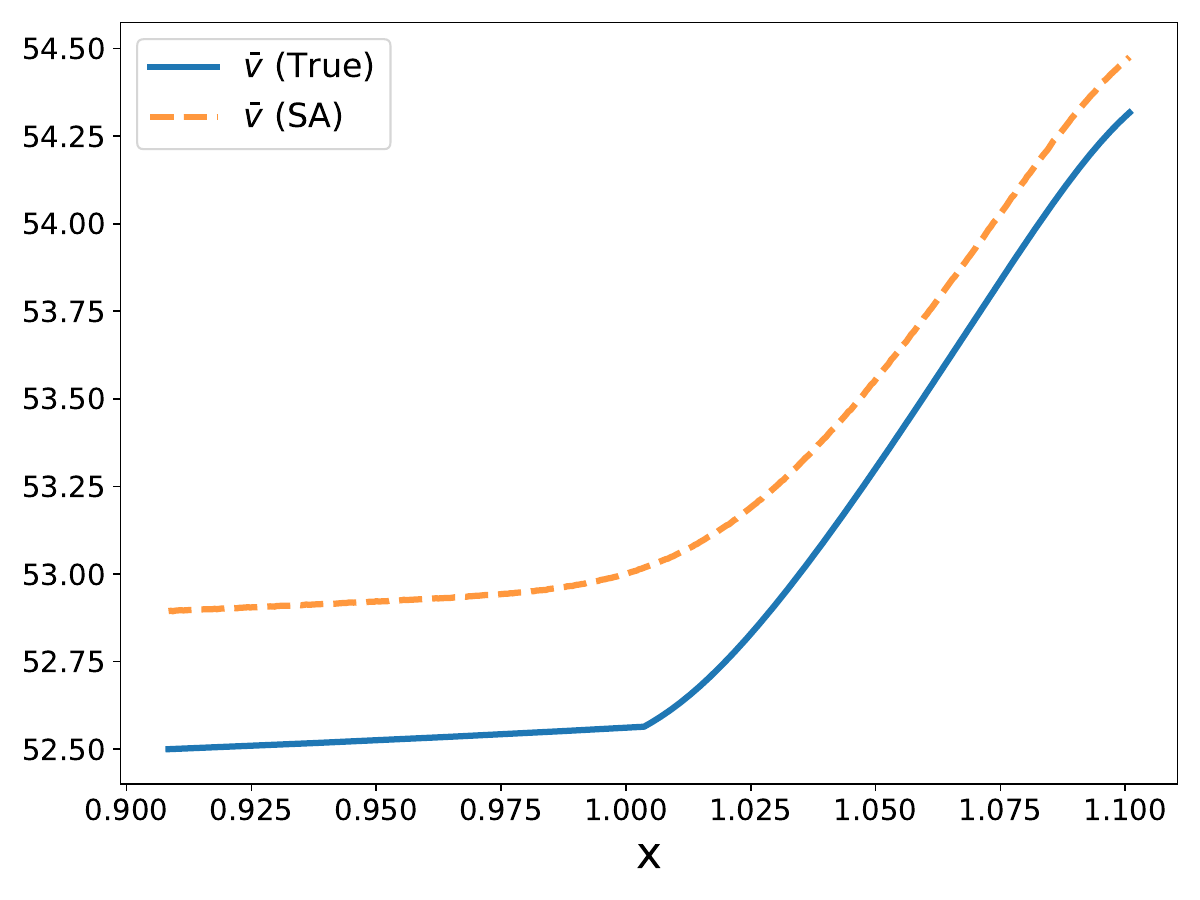}
    \caption{Evidence of max-of-means bias: ROOT-SA approximation of the fixed point $\bar v$. 
    Ground truth $\bar v$ (blue) vs. approximation using $\hat S$ (orange dashed).}
    \label{fig:rootsa_v}

    \centering
    \includegraphics[width=0.75\textwidth]{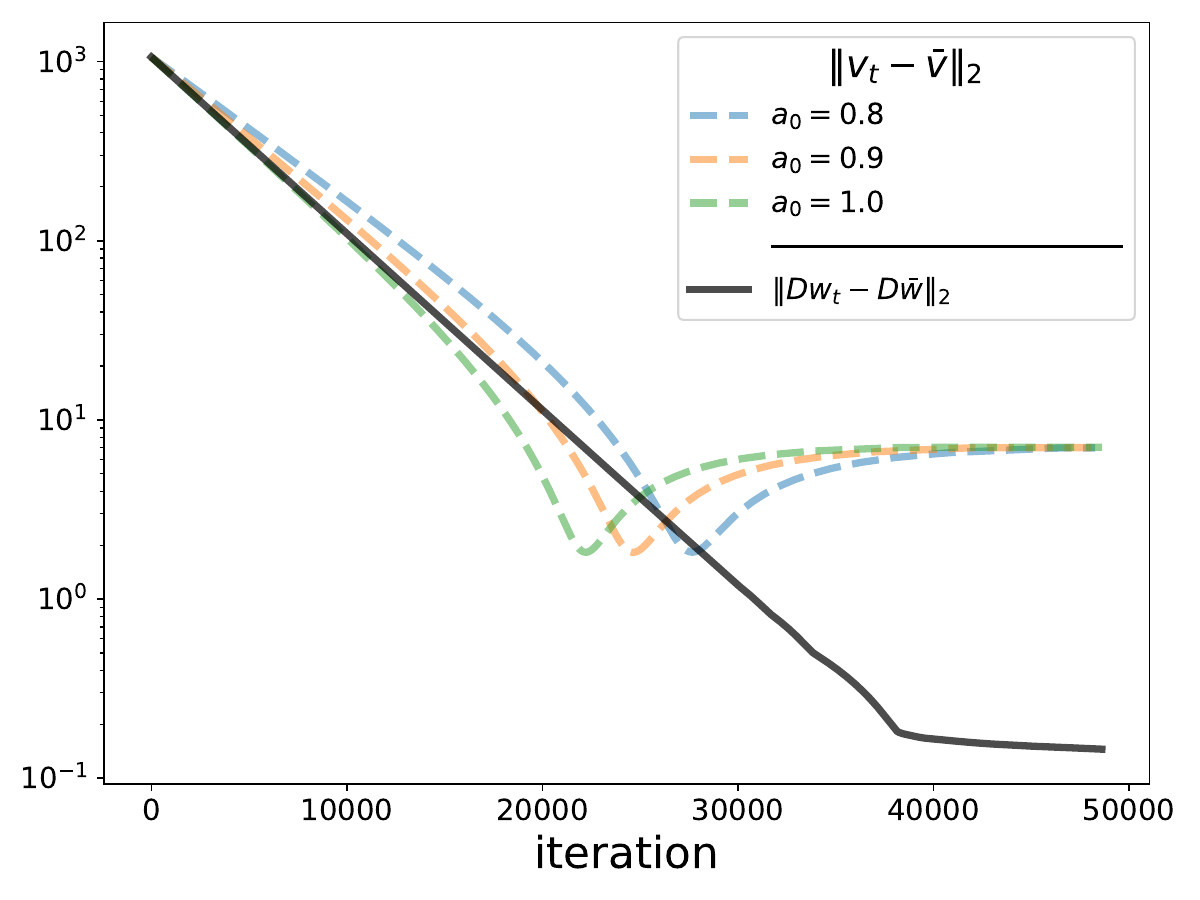}
    \caption{Convergence paths for $\hat S$ (dashed lines, various $\alpha_0$) and 
    $\hat T$ (black line, $\alpha_0 = 1.0$).}
    \label{fig:rootsa_v_path}
\end{figure}

\begin{figure}[ht]
    \includegraphics[width=0.8\textwidth]{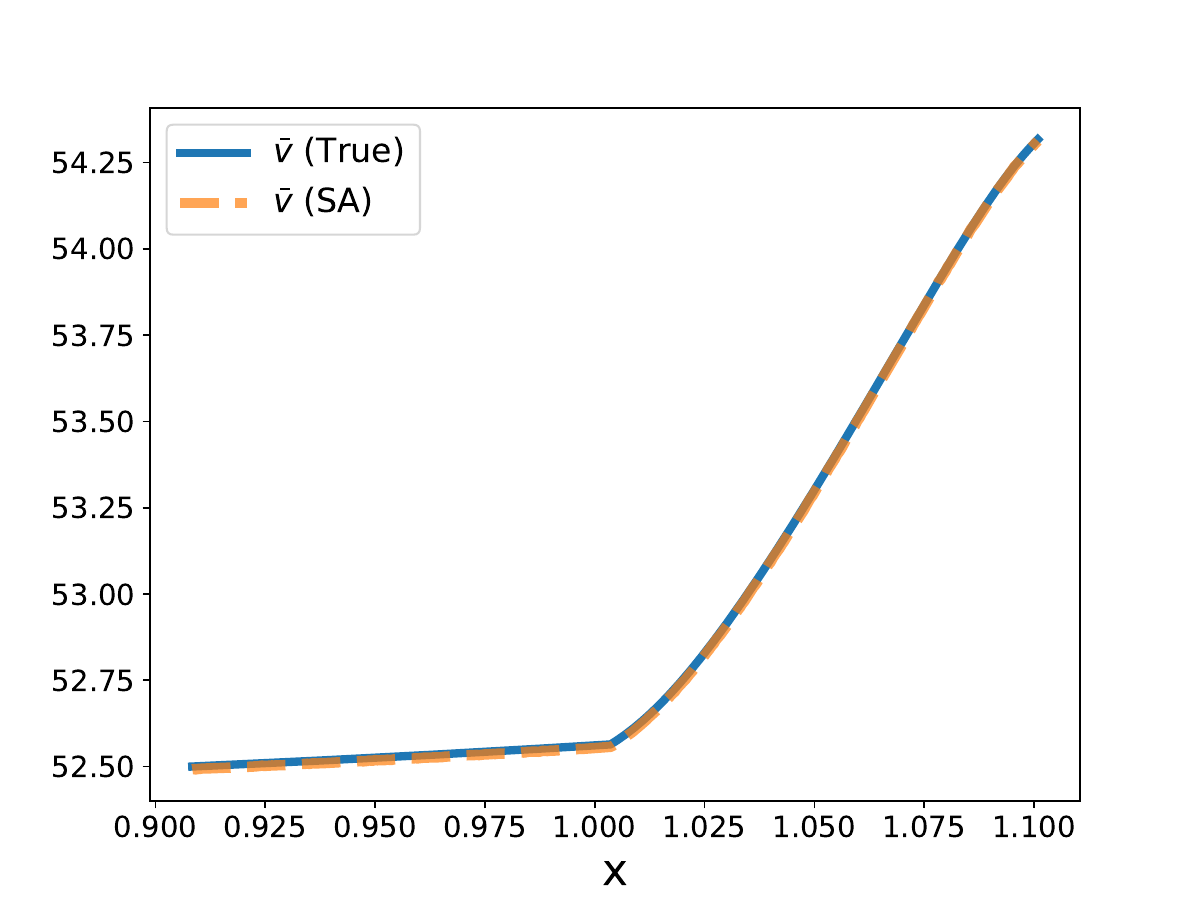}
    \caption{Comparison of the true equilibrium price function $\bar v$ calculated
    by iterating the operator $S$ and the ROOT-SA approximation $D w_M$.}
    \label{fig:rootsa_pi}
\end{figure}

\section{Upward Bias}\label{s:mc_approx}

In this section, we illustrate the claim that applying an approximation to 
equation \eqref{eq:eq_as} introduces a systematic upward bias. To see this, 
we apply the same ROOT-SA algorithm, replacing $\hat T$ with $\hat S$:
\begin{equation} \label{eq:S_hat}
    (\hat S v)(x) = \max_{i \in I} \beta\bigl[v(x_i') + d(x_i')\bigr], \qquad 
    x_i' \sim P_i \bigl(x,\cdot\bigr).
\end{equation}
By convexity of $y\mapsto\max_i y_i$ and Jensen's inequality,
\begin{equation*}
    (\EE \hat S v)(x)
    = \EE\max_i\beta\bigl[v(x_i')+d(x_i')\bigr]
    \geq \max_i\beta\EE\bigl[v(x_i')+d(x_i')\bigr]
    = (Sv)(x).
\end{equation*}

Therefore, the expected one-step operator $\EE[\hat S \cdot]$ is biased upward
relative to $S$. In our experiments, this one-step bias appears to accumulate,
pushing the iterates above the true fixed point, as illustrated in
Figure~\ref{fig:rootsa_v}, 
where we plot the solution obtained by applying ROOT-SA to $\hat S$. We can see that the 
approximation of $\bar v$ is systematically higher than the ground truth $\bar v$ 
(for more detailed discussion, see~\cite{hasselt2010double}). Similar results are also 
reported in~\cite{van2016deep}.

Moreover, we plot the convergence paths of applying ROOT-SA to $\hat S$ 
and $\hat T$, respectively, in Figure~\ref{fig:rootsa_v_path}. We can 
see that across different initial 
values of $\alpha_0$, the bias is systematic, and the deviation should
not be attributed to undertraining, given that all paths starting from different values 
of $\alpha_0$ converge to a higher value.
In contrast, the ROOT-SA algorithm applied to $\hat T$ converges to the 
true equilibrium price function $\bar v$ without the upward bias 
as shown in Figure~\ref{fig:rootsa_pi}.

\bibliographystyle{apalike}
\bibliography{qe_bib}

\end{document}